\documentclass[aps,pra, nopacs,twocolumn,twoside,floatfix,a4,superscriptaddress]{revtex4}
\usepackage{hyperref}
\usepackage{bbm}
\usepackage{amsmath, amssymb,amsthm}
\usepackage{color}
\usepackage{graphicx,epsfig}

\theoremstyle{plain}
\newtheorem{thm}{Theorem}
\newtheorem{lem}[thm]{Lemma}
\newtheorem{prop}[thm]{Proposition}

\theoremstyle{definition}
\newtheorem{defn}{Definition}

\usepackage{times} 

\definecolor{ngreen}{rgb}{0.1,0.5,0.1}

\newcommand{\rem}[1]{}

\newcommand{\tr}{{\textrm{tr}}}

\begin{document}

\title{Generalized Bell Inequality Experiments and Computation}

\author{Matty J. Hoban}
\affiliation{Department of Physics and Astronomy, University College London, Gower Street, London WC1E 6BT, United Kingdom.}
\affiliation{Department of Computer Science, University of Oxford, Wolfson Building, Parks Road, Oxford OX1 3QD, United Kingdom.}
\author{Joel J. Wallman}
\affiliation{School of Physics, The University of Sydney, Sydney, New South Wales 2006, Australia.}
\author{Dan E. Browne}
\affiliation{Department of Physics and Astronomy, University College London, Gower Street, London WC1E 6BT, United Kingdom.}

\date{\today}
\begin{abstract}
\noindent
We consider general settings of Bell inequality experiments with many parties, where each party chooses from a finite number of measurement settings each with a finite number of outcomes. We investigate the constraints that Bell inequalities place upon the correlations possible in a local hidden variable theories using a geometrical picture of correlations. We show that local hidden variable theories can be characterized in terms of limited computational expressiveness, which allows us to characterize families of Bell inequalities. The limited computational expressiveness for many settings (each with many outcomes) generalizes previous results about the many-party situation each with a choice of two possible measurements (each with two outcomes). Using this computational picture we present generalizations of the Popescu-Rohrlich non-local box for many parties and non-binary inputs and outputs at each site. Finally, we comment on the effect of pre-processing on measurement data in our generalized setting and show that it becomes problematic outside of the binary setting, in that it allows local hidden variable theories to simulate maximally non-local correlations such as those of these generalised Popescu-Rohrlich non-local boxes.
\end{abstract}

\maketitle

\section{Introduction}\label{sec1}

\noindent
Quantum mechanics is incompatible with a classical view of the world in many ways. In particular, quantum theory is incompatible with the assumption of local realism as there are entangled quantum states and measurements that lead to violations of a Bell inequality \cite{bell}. Consequently, Bell inequalities provide an extremely clear distinction between classical, local hidden variable (LHV), and non-classical (e.g., quantum) theories and so have been studied with great fervor since their discovery. 

Another motivation for the study of Bell inequalities has been from an informational point-of-view. Quantum information has brought new insight into the very nature of quantum mechanics and Bell inequalities have been used to gain insight into the information processing power of quantum mechanics. Bell inequality violations have been used to guarantee the security of quantum key distribution \cite{qkd}, generate randomness \cite{random}, give an advantage in communication complexity \cite{comm} and non-local games \cite{non-localgame}. Recently, connections have been made between quantum computing and a violation of a Bell inequality as giving some computational advantage \cite{anders, hoban, hoban2}. 

Despite their utility, there are many open questions about Bell inequalities. Many breakthroughs have been made by constructing Bell inequalities for many parties each with a choice of two measurements and two measurement outcomes \cite{chsh,ww,zb}. However, constructing Bell inequalities (that fully define LHV correlations) in general is NP-hard in the number of measurement choices and measurement outcomes at each site \cite{pitowski}.

Generalizations of the Bell inequality experiment away from the two measurement setting, two measurement outcome scenario have been studied and have led to the discovery of interesting phenomena \cite{vertesi,zukowski,beamsplitter,highdem,collins,cglmp}. For example, quantum violations can be greater \cite{cglmp} and more robust to experimental imperfections \cite{highdemloophole} in these general scenarios.

\begin{figure}
	\centering
		\includegraphics[width=0.49\textwidth]{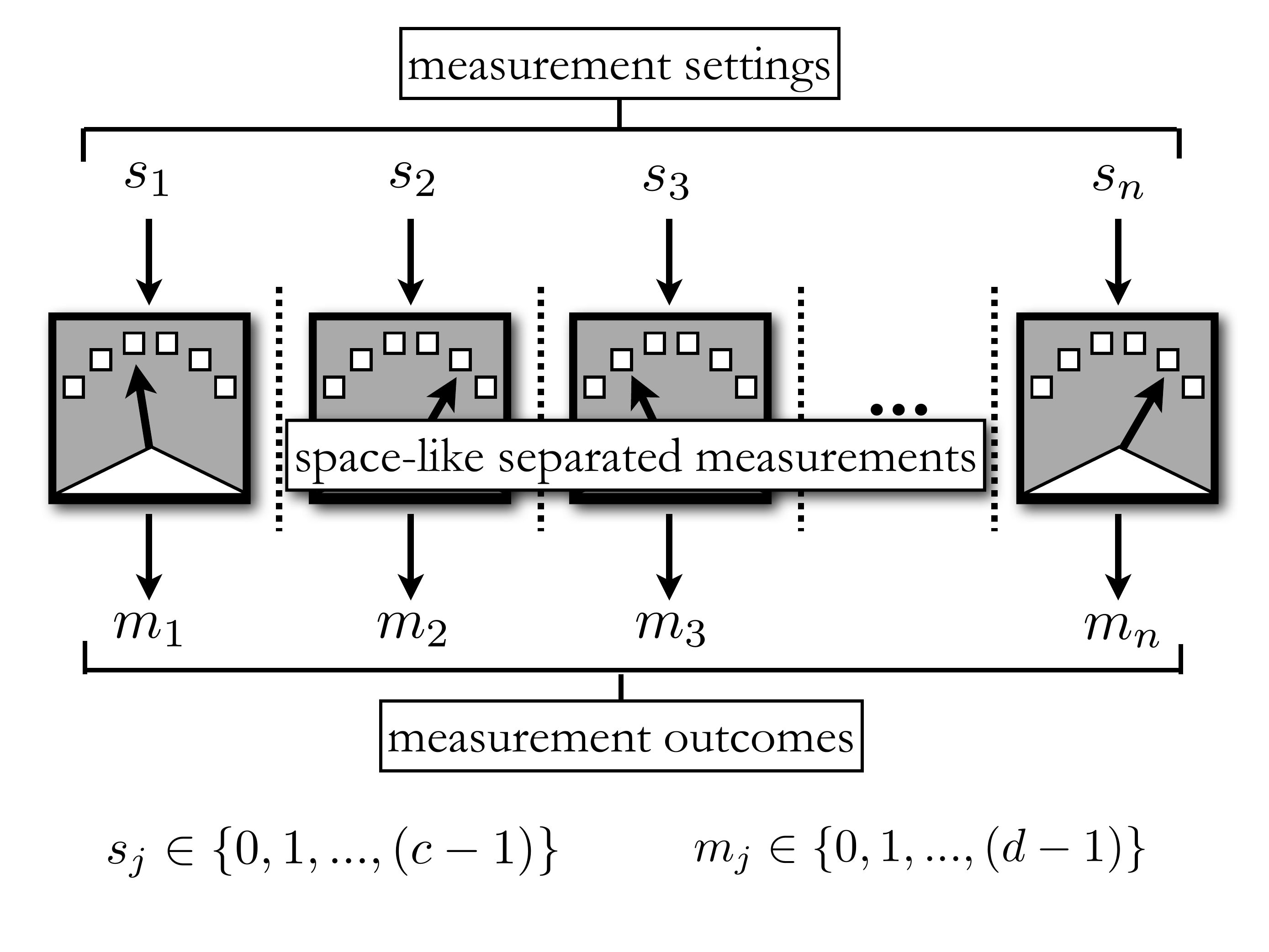}
			\caption{In a single run of a general Bell experiment, $n$ parties each make a measurement from $c$ possible choices, where each measurement has $d$ possible outcomes. Labelling the $j$th parties' measurement choice and outcome by $s_j$ and $m_j$ respectively, we can describe each run of the experiment with $n$-digit strings $\textbf{m}$ and $\textbf{s}$.}
			\label{fig:fig1}
\end{figure}

The Clauser-Horne (CH) inequality \cite{ch} also showed that the full probability distribution (including marginal probabilities) can be constrained in LHV theories and violated by quantum mechanics. However, quantum theory does not obtain the maximal violation of the CH inequality possible in theories that do not allow superluminal signaling (non-signaling for short), as demonstrated by Popescu-Rohrlich boxes \cite{pr}. As a result, much work has gone into characterizing the space of non-signaling probability distributions \cite{barrett2}. An approach based around convex polytopes has been particularly fruitful but the space of all non-signaling probability distributions remains complicated and often counter-intuitive, as demonstrated by the Guess Your Neighbor's Input non-local game \cite{gyni}. 

The convex polytope of non-signaling probability distributions in the CH inequality setting is well-studied and is an excellent platform for describing quantum correlations via an information theoretic, or physical principle \cite{cc,ic,mac,unc}. For more general settings, it is difficult to describe the space of non-signaling probabilities. To avoid this difficulty yet still obtain some intuition about the space of non-signaling probabilities, we consider the space of correlators, i.e., the space of probability distributions over joint outcomes.

For two-setting, two-outcome Bell inequalities, the correlations studied in Bell experiments can be described in terms of stochastic Boolean maps, e.g. probabilistic mixtures of Boolean functions from the two bits describing the measurement settings to the parity of the measurement outcomes, which is a single bit. This approach has proven well-suited to the study of many-body generalizations of the Clauser-Horne-Shimony-Holt (CHSH) inequality and led to a number of unifications and new insights \cite{anders, hoban, hoban2}.

In this paper, we generalize this approach by replacing stochastic Boolean maps with stochastic maps from the set of measurement settings to a single digit representing the parity of the output digits. Our generalization reveals new phenomena which do not occur in the simpler two-setting, two-outcome case.

Hand-in-hand with this ``computational'' description of correlations we shall use a geometric approach. The list of conditional probabilities describing the stochastic map can also be treated as a vector in a real vector space \cite{froissart,peres, pitowski}. In this picture, the valid sets of correlations which occur in a theory form a convex region in this space, which, in the case of LHV theories, is a polytope. A polytope can be defined as the intersection of the half-spaces that satisfy a set of linear inequalities, which, for the LHV polytope, are simply the facet Bell inequalities.

Very recently, two of us showed that for 2-setting 2-outcome Bell inequalities, computational expressiveness provides an elegant method for studying Measurement-based Quantum Computing \cite{hoban, hoban2}. We shall show how such approaches can be generalized, and that the switch from Boolean to more general maps allows more complex behavior than in the binary case.

In Section \ref{sec2}, we outline the general framework, fix notation and introduce the modular arithmetic which will be employed throughout the paper. In Section \ref{sec3} we construct Bell inequalities in this framework. In Section \ref{sec4} we study how methods from \cite{hoban} and \cite{hoban2} can be generalized and then summarize in Section \ref{sec5}. 

\section{Framework for general CHSH Tests}\label{sec2}

\noindent
The general scenario we consider is the $n$-partite Bell experiment shown in Fig. \ref{fig:fig1}, where $n>1$. For simplicity, we assume all space-like separated parties choose from $c$ possible measurement settings, each with $d$ possible outcomes, so we label an experiment by the variables $(n,c,d)$.

As in the standard CHSH experiment, we assume that each party's measurement setting is chosen randomly from a uniform distribution and is uncorrelated with the state of the system (see \cite{bell2,kofler,barrett,hall} for some consequences of dropping this assumption).

We label the $c$ possible measurements and $d$ possible measurement outcomes by digits $s_j\in\{0,\ldots,c-1\}\in\mathbb{Z}_{c}$ and $m_j\in\mathbb{Z}_{d}$ respectively. Therefore the $n$ settings and measurements for each repetition of the experiment are labeled by $n$ digit-strings $\textbf{m}\in\mathbb{Z}_{d}^{n}$ and $\textbf{s}\in\mathbb{Z}_{c}^{n}$ corresponding to all $n$ measurement outcomes and all $n$ settings respectively, where $\textbf{x}$ denotes the $n$-digit string such that the $j$-th digit is $x_j$.

To simplify the analysis we shall assume initially that both $c$ and $d$ are prime. We discuss features of the analysis in the non-prime case in appendix \ref{app2}. We will employ arithmetic in both modulo $c$ and modulo $d$. As we are considering correlators obtained by adding measurement outcomes, the majority of the arithmetic will be modulo $d$. Therefore we use the following convention. In all expressions, arithmetic will be modulo $d$, except in exceptional cases where the arithmetic is modulo $c$, which will be denoted with large square brackets and subscript $c$, i.e. $[\cdots]_c$.

After the data $\textbf{m}$ and $\textbf{s}$ has been collected from many repetitions of the experiments, conditional probabilities $p(\textbf{m}|\textbf{s})$ for each value of $\textbf{m}$ and $\textbf{s}$ can be calculated. In the initial CHSH paper, and in many-body generalizations, rather than studying this full probability distribution, one merely considers the statistics of the parity of the output $\textbf{m}$. Here, we shall take a similar course and study correlators
\begin{equation}\label{correlatorsdefined}
p(k|\textbf{s})=p(\sum_{j=1}^{n}m_{j}=k|\textbf{s}),
\end{equation}
where $p(\sum_{j=1}^{n}m_{j}=k|\textbf{s})=\sum_{\textbf{m}|\sum_{j=1}^{n}m_{j}=k}p(\textbf{m}|\textbf{s})$. When taking a sum over $m_{j}$ we are performing addition modulo $d$ but when taking sums of probabilities or describing anything that is not related to measurement settings or outcomes we are performing standard, natural arithmetic over the reals. This natural generalization of the CHSH-type correlators in \eqref{correlatorsdefined} also encompasses the well-known Collins-Gisin-Linden-Masser-Popescu (CGLMP) inequality for two parties \cite{cglmp}. These correlators have a well-defined role when considering the expectation values of joint measurements \cite{general}. Moreover, singling out this one parameter simplifies the problem by reducing the dimensionality of the problem, whilst still capturing interesting properties of the full distribution $p(\textbf{m}|\textbf{s})$.

In Section \ref{sec2a} we introduce the mathematical framework for describing these correlators, which employs a correspondence between conditional probabilities and stochastic maps. In Section \ref{sec2b} we apply this framework and derive the full family of correlators that are possible in LHV theories. Then, in Section \ref{sec2c} we consider correlators in more general non-signaling theories and derive generalizations of the Popescu-Rohrlich non-local box \cite{pr}.

\subsection{Correlators as Stochastic Maps}\label{sec2a}

\noindent
Before we begin our study of correlators, we need to introduce some notation and mathematical construction. The correlators $p(k|\textbf{s})$ are maps from $n$ digit-strings $\textbf{s}\in \mathbb{Z}_{c}^{n}$ to probability distributions over a single digit $k\in \mathbb{Z}_{d}$. Any stochastic map is a convex combination of deterministic processes, which in this case are functions $f:\mathbb{Z}_{c}^{n}\rightarrow\mathbb{Z}_{d}$.

Such functions take an ``input" $\textbf{s}\in \mathbb{Z}_{c}^{n}$ to a single ``output" $k\in\mathbb{Z}_{d}$. Since $\mathbb{Z}_{d}$ and $\mathbb{Z}_{c}^{n}$ are fields for $c$, $d$ prime, these functions can be represented as elements of a vector space. There are many bases one could choose to represent the function. Perhaps the simplest and cleanest basis is the set of Kronecker delta functions $\delta^{\textbf{y}}_{\textbf{s}}$, where $\textbf{y}\in \mathbb{Z}_{c}^{n}$, which equals 1 when $\textbf{y}=\textbf{s}$ and 0 otherwise. Any function can be represented as
\begin{equation}\label{genfunc}
f(\textbf{s})=\sum_{\textbf{y}\in\mathbb{Z}^{n}_{c}}f(\textbf{y})\delta^{\textbf{s}}_{\textbf{y}}.
\end{equation}
The $d^{c^n}$ possible coefficients correspond to $d^{c^n}$ functions. It will sometimes be convenient for the basis set to include the constant function. We shall then replace the delta function $\delta^{\mathbf{0}}_{s}$ where $\mathbf{0}_c$ is the all-zeros string with the constant function to write
\begin{equation}\label{eqn:f_delta}
f(\textbf{s})=\left(\sum_{\textbf{y}\in(\mathbb{Z}^{n}_{c}-\mathbf{\textbf{0}_{c}})}\epsilon_\textbf{y}\delta^{\textbf{s}}_{\textbf{y}}\right)+\alpha,
\end{equation}
where $\epsilon_\textbf{y}= (f(\textbf{y})-\alpha)$ and $\alpha\in\mathbb{Z}_{d}$. 

The natural expression for the computational power of correlators in a general theory is how close they can approximate an arbitrary function. Therefore we want to be able to rewrite \eqref{eqn:f_delta} as a polynomial. This can be done as the Kronecker delta $\delta^{\textbf{s}}_{\textbf{y}}$ can be written as a polynomial (modulo $c$),
\begin{eqnarray}\label{delta}
\delta^{\textbf{s}}_{\textbf{y}}&=&\prod_{j=1}^{n}\delta^{s_{j}}_{y_{j}}=\prod_{j=1}^{n}\left[1-(s_{j}-y_{j})^{(c-1)}\right]_{c},\nonumber \\
&=&\prod_{j=1}^{n}\left[1-\sum_{l=0}^{(c-1)}(-1)^{l}{{c-1}\choose{l}}(y_{j})^{l}(s_{j})^{c-(l+1)}\right]_{c}. \nonumber \\
\end{eqnarray}
The first line follows from Fermat's little theorem, as $a^{b-1} \equiv 1$ mod $b$ if $a$ is non-zero and $b$ is prime. To reach the second line we have used the binomial theorem which follows from the distributivity of modular arithmetic. 
This shows us that any function $f(\textbf{s})$ can be expressed as a polynomial with mixed modular arithmetic systems. Furthermore, when $d\ge c$ we can write the function solely in terms of modulo $d$ arithmetic since equation \eqref{delta} holds for any prime $\ge c$. 

In this paper, the following class of functions, which we term the ``$n$-partite linear functions'', will be crucial in the study of LHV theories.

\begin{defn}\label{defn:local-function}
A function $f(\textbf{s}):\mathbb{Z}^{n}_{c}\rightarrow\mathbb{Z}_{d}$ is an $n$-\textbf{partite linear function} if it can be expressed as
\begin{equation}
f(\textbf{s})=\sum_{j=1}^{n}g_{j}(s_{j}),
\end{equation}%
where addition is modulo $d$ and the $g_{j}(s_{j})$ are functions $\mathbb{Z}_{c}\rightarrow\mathbb{Z}_{d}$ of a single variable. If $f(\textbf{s})$ is not an $n$-partite linear function, we refer to it as a \textbf{non-}$n$\textbf{-partite linear function}.
\end{defn}

These functions will play a role analogous to that of the linear Boolean functions in \cite{hoban,hoban2}. The $n$-partite linear functions can also be written as
\begin{equation}\label{linear}
f(\textbf{s})=\alpha+\sum_{j=1}^{n}\sum_{a=1}^{(c-1)}\beta_{j,a}\delta^{s_{j}}_{a},
\end{equation}
with $\alpha$, $\beta_{j,a}\in\mathbb{Z}_{d}$. There are $d^{n(c-1)+1}$ $n$-partite linear functions. For $c$ being prime, the key feature of a non-$n$-partite linear function is the presence of cross-multiplicative terms in \eqref{delta} between different digits $s_{j}$ of $\textbf{s}$. 

For non-prime $c$, we cannot write delta functions in the neat polynomial of \eqref{delta} but in a more elaborate fashion (see appendix \ref{app2}) but we can still work in terms of the delta functions for non-prime $c$. Therefore, the definition of an $n$-partite linear function in \eqref{linear} applies for non-prime $c$. Thus for prime and non-prime $c$, the delta functions $\delta^{s_{j}}_{a}$ only singularly define a value of $f(\textbf{s})$ when there is only one non-zero element in $\textbf{s}$. For other values of $\textbf{s}$ there is addition modulo $d$ between the terms $\delta^{s_{j}}_{a}$. Equivalently then, for $f(\textbf{s})$ to be a non-$n$-partite linear function, there must be at least one delta function $\delta^{\textbf{s}}_{\textbf{y}}$ in \ref{eqn:f_delta} with $\textbf{y}$ being a digit-string with more than one non-zero element. A function $f(\textbf{s})$ can be written in terms of an $n$-partite linear part plus a non-$n$-partite linear part (with addition modulo $d$), and if the $n$-partite linear part is zero, then $f(\textbf{s})=0$ for all $\textbf{s}$ with only one non-zero element.

Finally, we note that the above treatment also applies in a more general setting where instead of setting digit-strings $\textbf{s}\in\mathbb{Z}_{c}^{n}$ we now have digit-strings $\textbf{s}\in\mathbb{Z}_{c_{1}}\times\mathbb{Z}_{c_{2}}\times...\times\mathbb{Z}_{c_{n}}=\bigoplus_{j=1}^{n}\mathbb{Z}_{c_{j}}$. That is, the number of possible settings, now labeled $c_{j}$ at each $j$th site could now be different from other sites and also non-prime. The string $\textbf{s}$ is now a direct sum (or Cartesian product) of the registers $\mathbb{Z}_{c_{j}}$ for each site. All delta functions as defined can be carried over to this general setting. Thus single site maps and $n$-partite linear functions are exactly the same as \eqref{linear} but for the $j$th site, $c$ is replaced with $c_{j}$. We now discuss a geometrical picture of the correlators $p(k|\textbf{s})$.

The correlators $p(k|\textbf{s})$ are stochastic maps and so must satisfy the positivity inequalities, i.e., the set of $c^{n}(d-1)$ linear inequalities $p(k|\textbf{s})\geq 0$, for all $k$, $\textbf{s}$. Furthermore, as some outcome always occurs for any choice of measurement settings $\textbf{s}$, the correlators must also satisfy the normalization equations, $\sum_{k\in\mathbb{Z}_d}p(k|\textbf{s})=1$ for all $\textbf{s}$. We then only have $(d-1)$ independent correlators for each value of $s$, so we omit the correlator $p(0|\textbf{s})$ and treat the remaining correlators as elements of a vector in a $c^{n}(d-1)$-dimensional real space with elements being $p(k|\textbf{s})$ for $k\neq 0$. 

The space of correlators satisfying the linear inequalities of the positivity conditions $p(k|\textbf{s})\geq 0$ and normalisation conditions in this reduced space, $\sum_{k\neq 0}p(k|\textbf{s})\leq 1$, is then a convex polytope labeled $\mathcal{P}$. It is a convex polytope as it is the intersection of the half-spaces defined by these linear inequalities. In any specific theory, the region of allowed correlators, $\mathcal{T}$, will be a convex subregion of $\mathcal{P}$. Equivalently, $\mathcal{P}$ can be defined as the convex hull of all deterministic correlators $p(k|\textbf{s})\in\{0,1\}$, for all ${k}$ and $\textbf{s}$; these deterministic correlators are the extreme points of $\mathcal{P}$.

If $\mathcal{T}$ has a finite number of extreme points (which is true for LHV theories but not for quantum theory), then $\mathcal{T}$ will be a convex polytope and so can be described as the intersection of the half-spaces defined by a set of ``facet-defining'' linear inequalities. Facet-defining inequalities are defined in terms of affinely independent points, where a set $\mathcal{S}$ of $K$ vectors $\vec{p}_{i}$, $\mathcal{S}=\{\vec{p}_{0},\vec{p}_{1},...,\vec{p}_{(K-1)}\}$ is \emph{affinely independent} if for every $\vec{p}_{k}\in\mathcal{S}$, the $(K-1)$ vectors in the set $\{\vec{p}_{i}-\vec{p}_{k}|\vec{p}_{i}\neq\vec{p}_{k}\}$ are linearly independent \cite{grun}.

\begin{defn}\label{defn:facet-defining}
A linear inequality is \textbf{facet-defining} for a convex polytope in $\mathbb{R}^{D}$ when at least $D$ affinely independent extreme points saturate the inequality (i.e. satisfy the equality of the linear inequality).
\end{defn}

Having constructed a framework for stochastic maps (and equivalently, conditional probabilities) we can immediately apply this to the correlators studied in this paper. Recall that we will focus our attention on correlators $p(k|\textbf{s})$ representing the probability distribution on $k=\sum_{j=1}^{n} m_{j}$, the sum modulo $d$ of the outputs. The correlators have the form of a conditional probability and thus a stochastic map, and therefore are characterized by vectors in the space $\mathcal{P}_{n,c,d}$. The space $\mathcal{P}_{n,c,d}$ represents the set of all correlators possible in principle, but for certain theories, not all correlators will be permitted. The focus of our study is then the regions of correlators which are accessible in LHV theories and in quantum theories. The region of LHV correlators is itself a polytope, and the facet-defining inequalities represent the facet-defining Bell inequalities\cite{pitowski, ww, zb,fine}. This is the geometrical picture of Bell inequalities that we will consider in the next subsection. We now study and characterize the region of correlators in LHV theories.

\subsection{Local Hidden Variable Theories}\label{sec2b}

\noindent
LHV theories are theories in which the measurement outcomes at each site can depend on the input setting of the measurement. The only other parameter the outcome depends on is an objective ``hidden variable'' $\lambda\in\Lambda$ which we assume is shared by all parties, and where $\Lambda$ is the (often continuous) set which defines the possible values for this variable. We do not assume that $\lambda$ is deterministic, but allow it to satisfy a probability distribution $p(\lambda)d\lambda$ over the space $\Lambda$ such that $\int_{\Lambda}p(\lambda)d\lambda=1$ . Without loss of generality \cite{fine}, we assume that apart from the probability distribution over $\lambda$ the internal workings of each measurement device is deterministic. This, via a standard argument, results in conditional probabilities for the set of outputs $\textbf{m}$ of the following form,
\begin{equation}
p(\textbf{m}|\textbf{s})=\int_{\Lambda}d\lambda p(\lambda)\prod_{j=1}^{n}p(m_{j}|s_{j},\lambda).
\end{equation}
Since $\lambda$ is the only non-deterministic element to consider, we can study the LHV correlations as the convex combination of the set of possible deterministic maps. We now prove that the map from measurement settings to correlators in any LHV theory is a probabilistic combination of $n$-partite linear functions.

\begin{thm}
In an $n$-party generalized Bell experiment, with $c$ settings and $d$ outputs for each measurement, the region $\mathcal{L}$ of correlators accessible in a LHV theory is the convex hull of the set of $n$-partite linear functions.
\end{thm}

\begin{proof}
The space $\mathcal{L}$ of correlators is a convex polytope, we only need to find the extreme points of the polytope, or deterministic correlators, and take their convex hull \cite{fine}. This means we just need to find the deterministic functions possible with LHV correlators. These deterministic functions correspond to adding the measurement outcomes of space-like separated sites, which, for a LHV theory, is deterministic if and only if the maps from the measurement setting $s_j$ to the measurement outcome $m_j$ at each site is deterministic, i.e., of the form $g_j(s_j)$ in Definition \ref{defn:local-function} for all $j$. Adding such functions modulo $d$ gives a $n$-partite linear function.
\end{proof}

\begin{figure}
	\centering
		\includegraphics[width=0.45\textwidth]{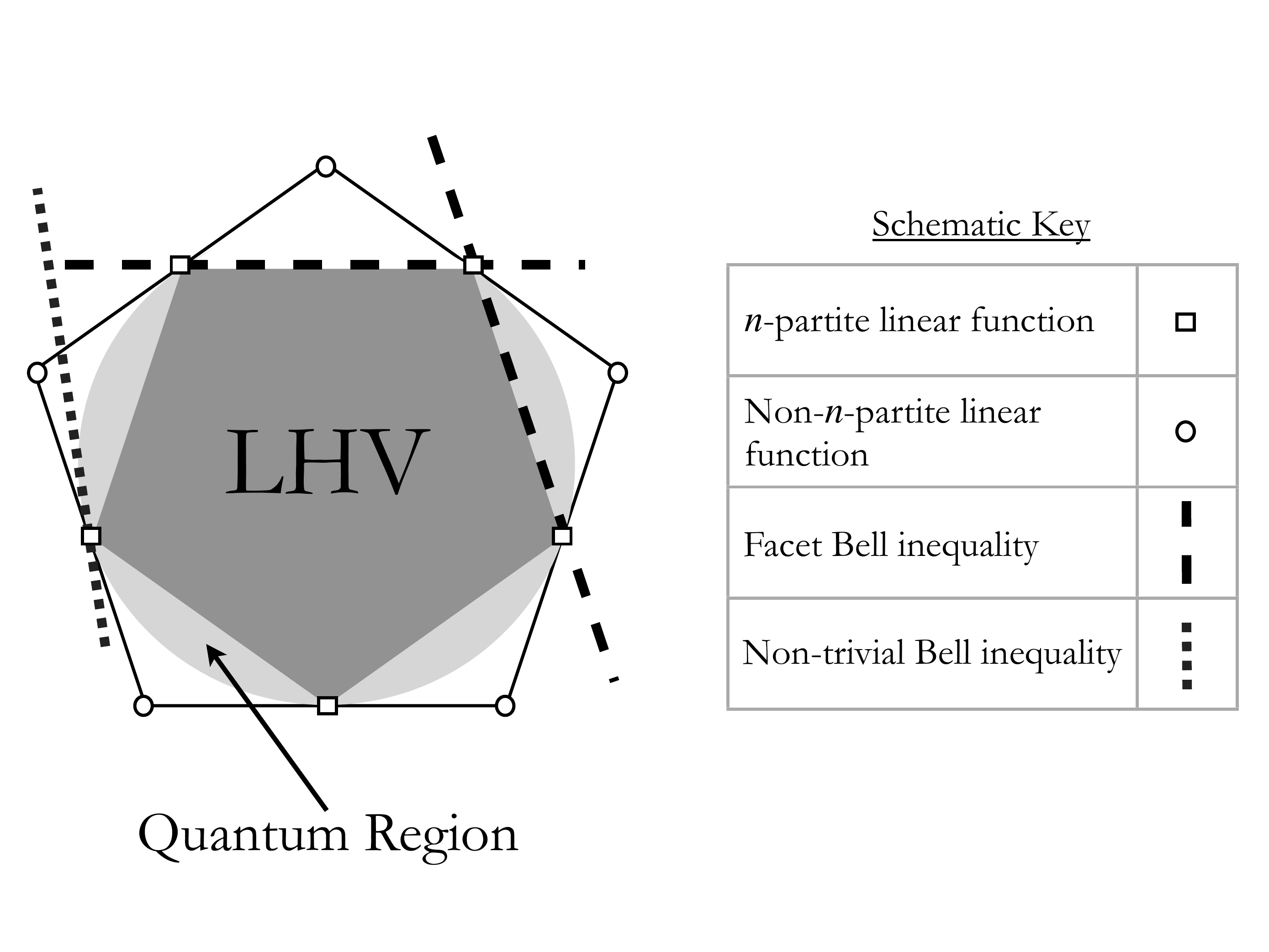}
			\caption{Schematic for correlators in LHV theories versus more general theories (including quantum mechanical correlators). The shapes represent the convex polytopes of general correlators and LHV correlators, i.e. the space $\mathcal{L}$. The points on each shape are the deterministic maps corresponding to the particular functions listed. We have shown the region of quantum correlators lying between the convex polytope of $n$-partite linear functions and the bigger convex polytope of all possible (including non-$n$-partite linear) maps. The facet Bell inequalities define the region $\mathcal{L}$ of LHV correlators whereas the non-trivial Bell inequalities (as mentioned in section \ref{sec3}) just bound the region. It is preferable but very hard \cite{pitowski} to find facet Bell inequalities rather than just non-trivial ones.}
			\label{fig:fig3}
\end{figure}

This theorem generalizes the theorem valid for the (n,2,2) case presented by two of us in \cite{hoban2}. It shows that $\mathcal{L}$ is compactly characterized by considering it in terms of the computational expressiveness of LHV correlators (see Fig. \ref{fig:fig3} for a schematic of the consequences of Theorem 1). In other words, if one considers the Bell experiment as a computation, LHV theories allow the computation of $n$-partite linear functions and nothing more. From this perspective, LHV theories have limited computational power, and thus a theory that allows correlators outside this region can be perceived as having a \emph{computational advantage} in a Bell experiment.

As mentioned above, there is a dual description of a convex polytope in terms of linear inequalities. The convex polytope of LHV correlators $\mathcal{L}$ is the intersection of the half-spaces described by the set of facet-defining Bell inequalities (or facet Bell inequalities for short). In Fig. \ref{fig:fig3} we illustrate the different types of Bell inequalities which one can construct, and the terminology to classify them. Bell inequalities that are equivalent to normalization or positivity conditions (i.e the boundaries of $\mathcal{P}$) are referred to as trivial Bell inequalities. There remain Bell inequalities that are neither necessarily facet-defining or trivial, which we call \emph{non-trivial} Bell inequalities; they are non-trivial because they indicate a separation between the polytopes $\mathcal{P}$ and $\mathcal{L}$ by bounding the latter. 

If a correlator is outside $\mathcal{L}$ then it must \emph{violate} one of the facet Bell inequalities. For example, it is well-known that some quantum correlators violate a Bell inequality \cite{bell}. Also if we observe a violation of a non-trivial inequality then we know that it is outside of $\mathcal{L}$. To recapitulate, for either non-trivial or facet Bell inequalities, \textit{we associate a violation with a computational advantage} and in Section \ref{sec3} we consider Bell inequalities from this computational perspective.

Before investigating Bell inequalities, it is worth considering the relation between $p(k|\textbf{s})$ and $p(\textbf{m}|\textbf{s})$. Since the latter is more general than the former, we are excluding interesting phenomena by only considering correlators. Examples of interesting phenomena in the full $p(\textbf{m}|\textbf{s})$ setting include the $I_{3322}$ inequality \cite{collins} which has interesting implications with regards to quantum correlations; another example is the Guess Your Neighbor's Input \cite{gyni} non-local game which results in some facet Bell inequalities that are not violated by any quantum correlators. In spite of these examples, we argue in the next subsection that the correlators $p(k|\textbf{s})$ capture many important phenomena of full non-signaling probability distributions. In particular, we show that particular deterministic correlators can be uniquely associated with a single non-signaling distribution, that being a many-body generalization of the PR non-local box \cite{pr}.

\subsection{Non-signaling theories}\label{sec2c}

\noindent
In this section, we shall consider general non-signaling theories. These are theories in which there are no constraints on correlations other than the no-signaling conditions. First, note that given any deterministic correlator $p(k=f(\textbf{s})|\textbf{s})$ where $k=\sum_{j} m_{j}$, there exists a non-signaling distribution such that $p(m_{j}|s_{j})=\frac{1}{d}$ and the non-signaling distribution results in the correlator. This is an immediate consequence of Theorem \ref{thm:generalized-pr}, proven below, but can be understood intuitively since every possible value of $m_{j}$ is consistent with one of more output strings $\textbf{m}$ where $\sum_{j} m_{j}=f(\textbf{s})$. One can then take a suitable mixture of these strings such that every output digit is maximally uncertain. Since a maximally uncertain output carries no information about $\textbf{s}$, the no-signaling condition is automatically satisfied.

We denote the space of probability distributions $p(\textbf{m}|\textbf{s})$ that satisfy the no-signaling conditions by $\mathcal{NS}_{n,c,d}$ (the indices will normally be omitted for clarity). This space is a convex polytope as it is the intersection of the half-spaces defined by the linear inequalities being the normalization, positivity and non-signaling inequalities (see \cite{barrett2} for details). Finding the vertices of this polytope is the dual problem of finding the facet-defining linear inequalities corresponding to some vertices (e.g. finding the facet Bell inequalities) which is an NP-hard problem. We now discuss some of the vertices of $\mathcal{NS}$ without resorting to this vertex enumeration problem in the following discussion. 

A famous example of a non-signaling distribution which cannot be simulated by LHV or quantum correlators for two spatially separated systems is the PR non-local box. The PR box is defined as a black box that results in measurement outcomes according to the probability distribution,
\begin{equation}
p(m_{1},m_{2}|s_{1},s_{2}) = 
\begin{cases} \frac{1}{2} & \text{if $m_{1}+ m_{2}=s_{1}s_{2}$,}
\\
0 &\text{otherwise.}
\end{cases}
\end{equation}
The PR box is described by the correlator $p(k=s_{1}s_{2}|\textbf{s})=1$. In other words, the PR box is associated with a correlator in $\mathcal{P}_{2,2,2}$. This correlator is a vertex of $\mathcal{P}_{2,2,2}$ corresponding to the deterministic property that the parity of outputs $m_{1}$ and $m_{2}$ is equal to the product of the inputs $s_{1}$ and $s_{2}$. Furthermore, the PR box is the only distribution in $\mathcal{NS}_{2,2,2}$ which is compatible with the correlator $p(k=s_{1}s_{2}|\textbf{s})=1$. In other words the vertex correlator $p(k=s_{1}s_{2}|\textbf{s})=1$ in $\mathcal{P}$ \emph{uniquely defines} a full probability distribution in $\mathcal{NS}$. 

In this section, we will show that this property is shared by many vertices in $\mathcal{P}$. These vertices are only compatible with a single probability distribution in $\mathcal{NS}$. In particular, for $n=2$, every correlator in $\mathcal{P}$ corresponding to a non-$n$-partite linear function is only compatible with one probability distribution in $\mathcal{NS}$.

We emphasize that the full set of vertices of $\mathcal{NS}$ for arbitrary $c$, $d$ and $n$ has never been characterized, and its numerical computation is an NP hard problem (the vertex enumeration problem) \cite{barrett2}. For example, in the $(3,2,2)$ setting, the space NS has been shown to have a complicated structure \cite{nspoly}. Therefore, we are able to capture a large amount of structure without resorting to optimization. 

Before we define this structure we need to introduce a new class of functions $f(\textbf{s}):\mathbb{Z}_{c}^{n}\rightarrow\mathbb{Z}_{d}$ that we will only use in this subsection to find a \emph{subset} of the vertices of $\mathcal{NS}$. We call these functions ``bi-partite linear" as they correspond to functions for which, if the observers gather into two groups at two space-like separated sites, the measurement outcome results from an $n$-partite linear function for $n=2$.

\begin{defn}\label{defn:bipartition}
A \textbf{bipartition}, $\{A,B\}$ of the set of $n$ parties is a division of the set $\{1,2,\ldots,n\}$ into two disjoint, non-empty sets $A$ and $B$ such that $A\cup B = \{1,2,\ldots,n\}$.
\end{defn}

\begin{defn}\label{defn:bi-local}
A function $f(\textbf{s}):\mathbb{Z}_{c}^{n}\rightarrow\mathbb{Z}_{d}$ is \textbf{bi-partite linear} if for any bipartition $\{A,B\}$, $f(\textbf{s})$ can be written as
\begin{equation}\label{bipartitelinear}
f(\textbf{s})=f^{A}(\textbf{s}^{A})+f^{B}(\textbf{s}^{B}),
\end{equation}
where $\textbf{s}^X$ is the $|X|$-digit string with entries $s^X_j = s_{X_j}$ and $f^{X}(\textbf{s}^{X})$ is a function on $\textbf{s}^{X}$.
\end{defn}

We now show that there exists a unique non-signaling distribution corresponding to a vertex $p(k=f(\textbf{s})|\textbf{s})$ of $\mathcal{P}$ if and only if $f(s)$ is not bi-partite linear.

\begin{thm}\label{thm:generalized-pr}
For every function $f(\textbf{s})$, the corresponding vertex $p(k=f(\textbf{s})|\textbf{s})=1$ of $\mathcal{P}$ is compatible with the non-signaling probability distribution
\begin{equation}\label{genbox}
p(\textbf{m}|\textbf{s}) = 
\begin{cases} d^{1-n} & \text{if $\sum_{j=1}^{n}m_{j}=f(\textbf{s})$,}
\\
0 &\text{otherwise.}
\end{cases}
\end{equation}
Furthermore, this is the only non-signaling distribution compatible with $p(k=f(\textbf{s})|\textbf{s})=1$ if and only if $f(\textbf{s})$ is not bi-partite linear.
\end{thm}

\begin{proof}
One can verify by inspection that this distribution is non-signaling (since all marginals are maximally uncertain) and satisfies $p(k=f(\textbf{s})|\textbf{s})=1$. It remains to be shown that the distribution in \eqref{genbox} is unique if and only if $f(\textbf{s})$ is not bi-partite linear. To prove the if statement, note that for any bi-partite linear function $f(\textbf{s})$, the non-signaling distribution
\begin{equation}\label{bigenbox}
p(\textbf{m}|\textbf{s}) = 
\begin{cases} d^{2-|A|-|B|}=d^{2-n} & \text{if $\sum_{j\in A}m_{j}=f^{1}(\textbf{s}^{A})$}
\\
& \text{and $\sum_{j\in B}m_{j}=f^{2}(\textbf{s}^{B})$,}
\\
0 &\text{otherwise,}
\end{cases}
\end{equation}
is compatible with $p(k=f(\textbf{s})|\textbf{s})=1$. We prove the only if statement in Lemma \ref{lem:if}.
\end{proof}

Note that local operations, which, for the Bell experiment setting, corresponds to adding functions $g'_j(s_j)$ to the measurement outcomes at each site, do not change whether or not a function is bi-partite linear. Since every function can be written in terms of an $n$-partite linear part and a non-$n$-partite linear part, local operations can create the $n$-partite linear part and so without loss of generality this part can be set to be zero. Therefore it is sufficient to consider functions $f(\textbf{s})$ such that $f(\textbf{t})=0$ for all $n$-digit strings $\textbf{t}\in\mathbb{Z}_{c}^{n}$ that have at most one non-zero digit.

We can separate the parties using any bipartition $A$ and $B$ and treat the sets of parties, $A$ and $B$, as single parties. To do this, we take the sum modulo $d$ of all outcomes on each side of the partition and set the number of settings on each side of the partition to be the product of the number of settings of the parties on that side of the partition (e.g., $c_A = \prod_{i\in A} c_i$). For a function that is not bi-partite linear, for all bipartitions into $A$ and $B$, the function $f:\mathbb{Z}_{c_{A}}\times\mathbb{Z}_{c_{B}}\rightarrow\mathbb{Z}_{d}$ is non-$n$-partite linear for $n=2$. Therefore it is sufficient to show that Theorem \ref{thm:generalized-pr} applies for $n=2$ and so it must apply for all bipartitions of $n$ parties as long as a function $f(\textbf{s})$ in \eqref{genbox} is not bi-partite linear. If the probability distribution in \eqref{genbox} applies for $n=2$ then it must apply for all possible bipartitions of $n$ parties thus generating the distribution \eqref{genbox} for all $n$.

\begin{lem}\label{lem:if}
For every function $f(s_{1},s_{2}):\mathbb{Z}_{c_{1}}\times\mathbb{Z}_{c_{2}}\rightarrow\mathbb{Z}_{d}$ that is non-$n$ partite linear for $n=2$, the only non-signaling distribution compatible with the corresponding vertex $p(k=f(s_{1},s_{2})|\textbf{s})=1$ in $\mathcal{P}$ is
\begin{equation}
p(m_{1},m_{2}|s_{1},s_{2}) = 
\begin{cases} d^{-1} & \text{if $m_{1}+m_{2}=f(s_{1},s_{2})$,}
\\
0 &\text{otherwise.}
\end{cases}
\end{equation}
\end{lem}

\begin{proof}
The condition $p(k=f(s_{1},s_{2})|s_{1},s_{2})=1$ for all $\textbf{s}=\{s_{1},s_{2}\}$ implies that for every value of $m_{1}$ in $p(m_{1},m_{2}|s_{1},s_{2})$, there exists a unique value of $m_{2}=f(\textbf{s})-m_{1}$. In computer science terminology we would call this a ``unique game". This immediately implies the equality for the following conditional distributions:
\begin{eqnarray}
p(m_{1}=x,m_{2}=f(\textbf{s})-x|\textbf{s})&=&p(m_{1}=x|\textbf{s})\nonumber \\
&=&p(m_{2}=f(\textbf{s})-x|\textbf{s}), \nonumber \\
\end{eqnarray}
for all $x\in\mathbb{Z}_{d}$. The no-signaling condition further implies that $p(m_{1}=x|\textbf{s})=p(m_{1}=x|s_{1})$ and
\begin{equation}\label{marginal}
p(m_{1}=x|s_{1})=p(m_{2}=f(\textbf{s})-x|s_{2}),
\end{equation}
which must be satisfied for all $\textbf{s}$ and all $x$. We will show that repeated application of \eqref{marginal} for varying $\textbf{s}$ allows us to prove that all non-marginal probabilities are equal provided that $f(\textbf{s})$ has a non-$n$-partite linear element. As discussed above, we can set $f(0,s_{2})=f(s_{1},0)=0$ for all $\textbf{s}=\{s_{1},s_{2}\}$. For such functions, repeatedly applying \eqref{marginal} gives
\begin{align}
p(m_2 = - x|s_2) &= p(m_1 = x|0) \nonumber\\
&= p(m_2 = - x|0) \nonumber\\
&= p(m_1 = x|s_1) \nonumber\\
&= p(m_2 = f(s_1,s_2) - x|s_2)
\end{align}
for all $x$. Repeated iteration implies 
\begin{equation}
p(m_2 = - x|s_2) = p(m_2 = \alpha f(s_1,s_2) - x|s_2)
\end{equation}
for all $\alpha\in\mathbb{Z}_d$. The function $f(\textbf{s})$ is non-$n$-partite linear so it must have at least one value of $\{s_{1},s_{2}\}$ where $f(s_{1},s_{2})$ is non-zero. Since $d$ is prime, $\alpha f(s_1,s_2)$ takes on all values in $\mathbb{Z}_d$, therefore the marginals are $p(m_2|s_2)=d^{-1}$ for all $m_2,s_2$. Applying equation \eqref{marginal} implies that $p(m_{1},m_{2}|s_{1},s_{2})=d^{-1}$ for all $\textbf{m}$ such that $m_{1}+m_{2}=f(\textbf{s})$.
\end{proof}

As we now show, for vertices $p(k=f(\textbf{s})|\textbf{s})=1$ of $\mathcal{P}$ such that $f(\textbf{s})$ is bi-partite linear, the existence of multiple non-signaling distributions implies that the vertex of $\mathcal{P}$ corresponds to a facet of $\mathcal{NS}$.

\begin{prop}\label{prop:NS-facets}
Every non-signaling distribution corresponding to a vertex of $\mathcal{P}$ lies on a facet of $\mathcal{NS}$.
\end{prop}

\begin{proof}
Any non-signaling distribution $p(\textbf{m}|\textbf{s})$ can be written as a convex combination $p(\textbf{m}|\textbf{s})=\sum_{E}p(E)p_{E}(\textbf{m}|\textbf{s})$ of the vertices $E$ of $\mathcal{NS}$, where $\sum_E p(E) = 1$. Therefore the associated correlator is
\begin{align}\label{eqn:prob-correlator}
p(k|\textbf{s}) &=\sum_{E}p(E)\delta^{\sum_{j=1}^{n}m_{j}}_{k}p_{E}(\textbf{m}|\textbf{s})	\nonumber\\
&=\sum_{E}p(E)p_{E}(k|\textbf{s}),
\end{align}
where $p_{E}(k|\textbf{s})=\delta^{\sum_{j=1}^{n}m_{j}}_{k}p_{E}(\textbf{m}|\textbf{s})$ is the correlator resulting from each extreme point of $\mathcal{NS}$.

As the non-signaling distribution results in a correlator that is a vertex of $\mathcal{P}$, $p(E)=0$ for all vertices of $\mathcal{NS}$ that do not result in the same correlator. Denote by $\mathcal{E}$ the set of vertices of $\mathcal{NS}$ that correspond to the correlator and let $\mathcal{F}$ be the region of convex combinations of the elements of $\mathcal{E}$. Any element in $\mathcal{F}$ must correspond to the correlator by \eqref{eqn:prob-correlator}.

$\mathcal{F}$ is not a facet of $\mathcal{NS}$ if and only if there exists a convex combination of the vertices not in $\mathcal{E}$ that intersects $\mathcal{F}$. However, if there exists such a convex combination, then this would give a convex combination of correlators in $\mathcal{P}$ that is equal to a vertex of $\mathcal{P}$, which is a contradiction.
\end{proof}

The uniqueness of the probability distributions in Theorem \ref{thm:generalized-pr} implies that each vertex $p(k=f(\textbf{s})|\textbf{s})=1$ of $\mathcal{P}$ such that $f(\textbf{s})$ is not bi-partite linear corresponds to a unique vertex of $\mathcal{NS}$ (i.e., the ``facet'' in Proposition \ref{prop:NS-facets} collapses to a single point).

\begin{prop}\label{cor:NS-vertices}
Every probability distribution of the form \eqref{genbox} is a vertex of $\mathcal{NS}$ if and only if $f(\textbf{s})$ is not a bi-partite linear function.
\end{prop}

\begin{proof}
As there is a unique non-signaling distribution corresponding to the correlator $p(f(\textbf{s})|\textbf{s})=1$, there is at most one $E$ such that $p(E)\neq0$ in \ref{eqn:prob-correlator}.
\end{proof}

The structure of the high dimensional space of all possible non-signaling distributions can be partly revealed by considering the space of correlators, which are of relatively low dimension. Note that the difficulty in ascertaining the remaining structure of $\mathcal{NS}$ comes from the difficulty of finding the set of vertices that generate vertices of $\mathcal{P}$ that correspond to bi-partite linear functions and determining whether this gives a complete list of vertices. Furthermore, much of the structure of LHV theories in $\mathcal{NS}$ is reduced to a small number of vertices of $\mathcal{P}$. In this way, the space of correlators efficiently encapsulates vital information about all non-signaling probability distributions. 

There is a connection between the notion of ``true $n$-party non-locality'' \cite{svet, barrett2} and the vertices of $\mathcal{NS}$ defined by Theorem \ref{thm:generalized-pr}. The bi-partite linear functions corresponding to vertices in $\mathcal{P}$ only reveal ``non-locality'' (i.e. correlations not resulting from LHV theories) between a subset of all parties, but no non-locality across at least one bipartition. However, functions that are not bi-partite linear result in correlations that are non-local across all partitions. 

\subsection{Summary}\label{sec2d}

\noindent
In this section we have constructed a framework for Bell correlators and showed that these correlators capture important properties of the space of non-signaling probability distributions. These correlators also capture the properties of LHV theories in a natural information theoretic description, i.e. that of computational expressiveness. This work generalizes previous work of \cite{hoban,hoban2} and shows where the structure in these previously studied settings emerges from. 

Given this framework, we want to look at the correlators that are achievable with quantum theory. It is well-known that there exists quantum mechanical correlators that cannot be described by correlators resulting from an LHV theory, and this is demonstrated by the violation of a Bell inequality \cite{bell}. Therefore we now devote our attention to the Bell inequalities.

\section{Bell inequalities}\label{sec3}

\noindent
Bell inequalities allow us to characterize the polytope of LHV correlators. The normalization and positivity inequalities defined in Section \ref{sec2b} define part of $\mathcal{L}$, the LHV polytope. However, these inequalities are also satisfied by all possible correlators $p(k|\textbf{s})$ and define the polytope $\mathcal{P}$ of all correlators. These inequalities then trivially bound $\mathcal{L}$ as we cannot show a separation between $\mathcal{P}$ and $\mathcal{L}$. We want to demonstrate such a separation and non-trivially bound $\mathcal{P}$, which we can do by constructing Bell inequalities of the form
\begin{equation}\label{bell}
\sum_{\textbf{s},k}\omega_{k,\textbf{s}}p(k|\textbf{s})\leq \gamma_{\mathcal{L}},
\end{equation}
for all $p(k|\textbf{s})\in\mathcal{L}$, where $\omega_{k,\textbf{s}}$ and $\gamma_{\mathcal{L}}$ are real numbers. We call $\gamma_{\mathcal{P}}$ the maximum value of the sum on the left-hand-side over correlators in $\mathcal{P}$. For the Bell inequality to be non-trivial, we require $\gamma_{\mathcal{L}}<\gamma_{\mathcal{P}}$.

If a correlator exists outside the LHV polytope, then it violates at least one facet Bell inequality. Finding facet Bell inequalities involves an optimization of the parameters $\omega_{k,s}$ in \eqref{bell} and in general is an NP-hard problem \cite{pitowski}. The problem is a facet enumeration problem and there are algorithmic packages that can perform this optimization such as Polymake \cite{polymake}. These optimization algorithms are often used to characterize the non-signaling \cite{barrett2} and LHV polytopes \cite{collins} in full probability distribution settings. We used Polymake to find the facet Bell inequalities for a number of small $(n,c,d)$ settings. The number of facet Bell inequalities for each setting is listed in Tab. \ref{tab:tab1}, which shows that the number of inequalities grows rapidly in the size of the problem.

\begin{table}
\begin{center}
 \begin{tabular}{| c | c | c | c | c |}
 \hline
 n & c & d & \# Vertices &\# Facet Bell inequalities \\ \hline
 2 & 2 & 2 & 8 & 16 \\ 
 2 & 2 & 3 & 27 & 66 \\
 2 & 2 & 5 & 125 & 1020 \\
 3 & 2 & 2 & 16 & 256 \\
 3 & 2 & 3 & 81 & 125,412 \\
 2 & 3 & 2 & 32 & 90 \\
 \hline
 \end{tabular}
\end{center}
\caption{A table of number of facet Bell inequalities for each setting $(n,c,d)$ and the number of vertices for the LHV polytope.}
\label{tab:tab1}
\end{table}
 
For the case of $(n,2,2)$, the LHV polytope is a hyperoctahedron \cite{ww,zb}. In general though, there is no obvious geometrical structure or connection between the numbers of vertices and the numbers of facet inequalities. It is quite apparent that finding the facet Bell inequalities is no easy task. 

In this section we will use the computational insight gained in Section \ref{sec2} to construct Bell inequalities. We first find a simple way to generate non-trivial Bell inequalities that necessarily bound the LHV polytope. As we shall discuss further in Section \ref{sec4}, these non-trivial inequalities also have a nice cross-over with the structures in so-called \emph{non-local games}. We then describe some of the facet Bell inequalities and show that they too have a computational nature. 

Finally, in the last part of this section, we discuss the quantum violation of facet Bell inequalities. We show that there is a violation for all facet Bell inequalities for $n=2$ in Tab. \ref{tab:tab1}. We comment on the quantum states that \emph{maximally} violate these Bell inequalities and find that in few instances is it a maximally entangled state. This supports the view that entanglement and a violation of a Bell inequality are not synonymous concepts.

\subsection{Non-trivial Bell inequalities}\label{sec3a}

\noindent
One might ask what is the point of looking for Bell inequalities that are not facet-defining? Firstly, bounding the set of LHV correlators still creates a non-trivial subregion of a correlator space. Secondly, we can still demonstrate non-classical behavior without the computational difficulty that the facet-defining condition creates. Finally, despite not being facet-defining, non-trivial Bell inequalities can have a role in particular applications such as in non-local games or Measurement-based Quantum Computing (MBQC) \cite{hoban2}, as non-trivial Bell inequalities might capture computational power. We will elaborate on this final point in Section \ref{sec4}. 

In this subsection, we will discuss non-trivial Bell inequalities as being inequalities in the space $\mathcal{P}'$ of all possible correlators $p(k|\textbf{s})$, including those for $k=0$. In this space, the normalisation conditions become $\sum_{k}p(k|\textbf{s})=1$ for all $\textbf{s}$, which defines a bounded hyperplane in $\mathbb{R}^{dc^{n}}$. The space of correlators $\mathcal{P}'$ is a convex polytope defined by the positivity conditions within this hyperplane. We regard the region of LHV correlators, $\mathcal{L}$, as a sub-region in both $\mathcal{P}$ and $\mathcal{P}'$ as one can always go between the two spaces by disregarding $p(0|\textbf{s})$ and applying the normalisation inequalities. In the rest of the paper, including the discussion on facet Bell inequalities, we always discuss correlators in the space $\mathcal{P}$.

Given the motivation for finding non-trivial Bell inequalities we now present a simple way of generating non-trivial Bell inequalities. We begin by considering the CHSH inequality \cite{chsh} as an example. The CHSH inequality can be written as
\begin{equation}\label{chsh}
\sum_{s_{1},s_{2}}\sum_{k=0}^{1}\delta^{k}_{s_{1}s_{2}}p(k|s_{1},s_{2})\leq 3,
\end{equation}
which is \eqref{bell} with $\gamma_{\mathcal{L}}=3$ and $\omega_{k,s}=\delta^{k}_{s_{1}s_{2}}$. The CHSH inequality is also a facet Bell inequality for this setting.

By convexity, it is only necessary to consider the vertices of the LHV polytope to obtain the bound $\gamma_{\mathcal{L}}=3$. In this case, the vertices correspond to the linear boolean functions of $s_1$ and $s_2$ \cite{hoban}. For these vertices the correlators are then $p(k|s_{1},s_{2})=\delta^{k}_{g(s_{1},s_{2})}$ where $g(s_{1},s_{2})$ are the linear Boolean functions on $\{s_{1},s_{2}\}$. The sum in \eqref{chsh} is then $\sum_{s_{1},s_{2}}\delta^{g(s_{1},s_{2})}_{s_{1}s_{2}}\leq 3$ as the functions $g(s_{1},s_{2})$ overlaps with the function $s_{1}s_{2}$ for at most $3$ values of $s$. An example of a linear function achieving this overlap is $g(s_{1},s_{2})=0$. On the other hand, if a correlator $p(k=f(\textbf{s})|s_{1},s_{2})$ achieves the map $f(\textbf{s})=s_{1}s_{2}$ deterministically, then it achieves a value of $\gamma_{\mathcal{P}}=4$ for the sum in \eqref{chsh}. This inequality highlights the inability for LHV correlators to evaluate non-linear Boolean functions deterministically. 

For all settings $(n,c,d)$, we can generalize the above concepts from the Boolean algebra for $(2,2,2)$ to our more general framework of maps $f:\mathbb{Z}_{c}^{n}\rightarrow\mathbb{Z}_{d}$. We construct non-trivial Bell inequalities in the following way:
\begin{equation}\label{genchsh}
\sum_{\textbf{s}}\sum_{k=0}^{d-1}\delta^{k}_{f(\textbf{s})}p(k|\textbf{s})\leq \underset{g(\textbf{s})}{\textrm{sup}}\sum_{\textbf{s}}\delta^{f(\textbf{s})}_{g(\textbf{s})},
\end{equation}
where the supremum is over the set of $n$-partite linear functions and $f(\textbf{s})$ is a non-$n$-partite function of $\textbf{s}$. Therefore there are $N=d^{c^{n}}-d^{n(c-1)+1}$ non-trivial Bell inequalities of the form \eqref{genchsh}. 

We can prove the inequality in \eqref{genchsh} using the same arguments as with the CHSH inequality. By convexity we only need to consider the vertices of $\mathcal{L}$ which correspond to the $n$-partite linear functions $g(\textbf{s})$. The sum in \eqref{genchsh} is then $\sum_{\textbf{s}}\sum_{k=0}^{d-1}\delta^{k}_{f(\textbf{s})}\delta^{k}_{g(\textbf{s})}$. We then take the supremum over all $n$-partite functions $g(\textbf{s})$ to obtain the right-hand-side of \eqref{genchsh}.

Whilst these inequalities might seem contrived, they include some interesting Bell inequalities. For example, the Svetlichny inequality \cite{svet} for three parties takes the form of \eqref{genchsh}, i.e.,
\begin{equation}
\sum_{s_{1},s_{2},s_{3}}\delta^{k}_{s_{1}(s_{2}+ s_{3})+ s_{2}s_{3}}p(k|s_{1},s_{2},s_{3})\leq 6,
\end{equation}
which is \eqref{genchsh} for $f(\textbf{s})=s_{1}(s_{2}+s_{3})+s_{2}s_{3}$. The Svetlichny inequality is not a facet Bell inequality for the LHV polytope $\mathcal{L}$ in the $(3,2,2)$ setting \cite{ww} but it still captures very interesting phenomena. {In particular, the Svetlichny inequality captures correlations that are consistent with LHV theories if one averages over the measurement settings and outcomes for any one of the parties.}

The non-trivial Bell inequalities \eqref{genchsh} can be seen to truly capture the inability of LHV theories to evaluate non-$n$-partite linear functions via correlators. A way of emphasizing this is to consider these inequalities as a ``non-local game'' \cite{non-localgame}. That is, given an input $\textbf{s}$ to $n$ parties that do not communicate, what is the maximum probability of the parties producing outcomes $m$ such that $\sum_{j=1}^{n}m_{j}=f(\textbf{s})$ for some function $f(\textbf{s})$ using different strategies (or resources such as shared randomness or entanglement). A referee chooses the input and then obtains all outputs $m$ before computing $\sum_{j=1}^{n}m_{j}$.

If the referee's choice of input is uniformly random, then the inequality in \eqref{genchsh} captures the mean success probability of an LHV strategy performing the function $f(\textbf{s})$; one just needs to divide both sides by $c^{n}$ to obtain a proper probability. If one imagines that the choice of input $\textbf{s}$ is not chosen at random but instead with some probability $p(\textbf{s})\geq 0$ such that $\sum_{\textbf{s}}p(\textbf{s})=1$ then we can modify \eqref{genchsh} to be:
\begin{equation}\label{nontrivial}
\sum_{\textbf{s}}p(\textbf{s})\sum_{k=0}^{d-1}\delta^{k}_{f(\textbf{s})}p(k|\textbf{s})\leq \underset{g(\textbf{s})}{\textrm{sup}}\sum_{\textbf{s}}p(\textbf{s})\delta^{f(\textbf{s})}_{g(\textbf{s})},
\end{equation}
to get the mean success probability of performing the function $f(\textbf{s})$ with LHV correlators given a distribution $p(\textbf{s})$ on inputs. We now present the following proposition that shows an infinite class of non-trivial Bell inequalities can be easily generated.
\newline

\begin{prop}
Any Bell inequality of the form \eqref{nontrivial} is non-trivial whenever $p(\textbf{s})\neq0$ for all $\textbf{s}$.
\end{prop}

\begin{proof}
For any $n$-partite linear function, 
\begin{align}
\sum_{\textbf{s}}p(\textbf{s})\delta^{f(\textbf{s})}_{g(\textbf{s})} &< \sum_{\textbf{s}}p(\textbf{s})  = 1
\end{align}
as there must be a value of $\textbf{s}$ such that $\delta^{f(\textbf{s})}_{g(\textbf{s})}=0$. Therefore the right-hand-side of \eqref{nontrivial} is strictly less than $1$ for all correlators in $\mathcal{L}$. However, for the correlator $p(k=f(\textbf{s})|\textbf{s})=1$ in $\mathcal{P}$, the left-hand side of \eqref{nontrivial} is exactly 1.
\end{proof}

The beauty of these non-trivial Bell inequalities is that they are easily generated as one can easily describe the $n$-partite linear functions for any setting $(n,c,d)$, yet they still bound the region of LHV correlators. They also explicitly capture the computational aspect of a Bell inequality experiment: the value of the left hand side of \eqref{nontrivial} gives the mean success probability of performing a particular function. 

We now consider the facet Bell inequalities that completely characterize the LHV polytope for small values of $(n,c,d)$. Whilst these latter inequalities do not have the immediacy of the non-trivial Bell inequalities and are difficult to calculate, they are still the optimal tool for demonstrating non-classical correlations. 

\subsection{Facet Bell inequalities}\label{sec3b}

\noindent
The facet Bell inequalities satisfy the facet-defining condition for LHV correlators and so if a correlator is outside the region of LHV correlators it must necessarily violate one of the facet Bell inequalities. In Tab. \ref{tab:tab1} we listed the number of facet Bell inequalities for a few settings that could be computed using the software Polymake. Included in the number of facet Bell inequalities are the $c^{n}$ normalization and $(d-1)c^{n}$ positivity inequalities that define the polytope $\mathcal{P}$. Despite these $dc^{n}$ inequalities, there are still a significant number of inequalities remaining. Also note we are again discussing correlators in the space $\mathcal{P}$ of all correlators except $p(0|\textbf{s})$ and $\mathcal{L}$ is the sub-region of LHV correlators in this space.

We can reduce the number of facet Bell inequalities by considering relabelings of the measurement choices, outcomes and parties, which we refer to as symmetries. If a Bell inequality can be changed into another inequality by such an operation, we refer to them as equivalent, or as members of the same symmetry class. Explicitly, the symmetries are
\begin{enumerate}
\item permutations of parties - $\{s_{i},s_{j},...,s_{n}\}\rightarrow\{s_{i'},s_{j'},...,s_{n'}\}$ where $k'=\sigma(k)$ is an element of the permutation group of order $n$
\item relabeling of measurement settings - $s_{j}\rightarrow s_{j}+\alpha_j$ for some $\alpha_j\in\mathbb{Z}_{c}$
\item relabeling of measurement outcomes - $m_{j}\rightarrow m_{j}+\beta(s_{j},j)$ where $\beta(s_{j},j)\in\mathbb{Z}_{d}$.
\end{enumerate}
The $n$-partite linear functions are closed under all of these operations. Using the facet-defining condition, the vertices of $\mathcal{L}$ that saturate a facet Bell inequality must be equivalent to another set of vertices in $\mathcal{L}$ by these operations that saturate another facet Bell inequality. Hence we can group together the facet Bell inequalities into these symmetry classes.

There are $n!$ permutations of $n$ parties and $c^{n}$ ways of relabeling measurement settings. Since for each input $s_{j}$ we add a value $\beta(s_{j})$, for each input $\textbf{s}$, $\beta(\textbf{s})=\sum_{j=1}^{n}\beta(s_{j})$ is added to $\sum_{j=1}^{n}m_{j}$. There will be at most $d^{cn}$ values of $\beta(\textbf{s})$.

We constructed a search algorithm to find all the symmetry classes of facet Bell inequalities for each setting $(n,c,d)$. In Tab. \ref{tab:tab2} we list the number of symmetry classes for each of the settings in Tab. \ref{tab:tab1}, which is dramatically fewer than the total number of facet Bell inequalities. One of the symmetry classes for each setting is the class of normalization and positivity conditions, leaving all the facet Bell inequalities that can be violated by LHV correlators outside of $\mathcal{L}$. We have therefore omitted the class of these positivity and normalization inequalities from Tab. \ref{tab:tab2} in order to leave the inequalities that can be violated.

We now look at a few particular families of facet Bell inequalities for the bipartite and tripartite setting based on the CGLMP inequality. We briefly discuss the $(2,3,2)$ setting and show how the CHSH inequality is essentially the only relevant facet inequality for this setting.

\begin{table}
\begin{center}
 \begin{tabular}{| c | c | c | c |}
 \hline
 n & c & d &\# Symmetry classes \\ \hline
 2 & 2 & 2  & 1 \\ 
 2 & 2 & 3  & 1 \\
 2 & 2 & 5  & 4 \\
 3 & 2 & 2  & 4 \\
 3 & 2 & 3  & 62 \\
 2 & 3 & 2  & 1 \\
 \hline
 \end{tabular}
\end{center}
\caption{The number of symmetry classes for each setting $(n,c,d)$. We have excluded the symmetry class of all positivity and normalisation inequalities from this number.}
\label{tab:tab2}
\end{table}

\textit{Bipartite Two-Setting Facet Bell Inequalities} - A symmetry class of note for $n=c=2$ is the class of facet Bell inequalities that are equivalent to the CGLMP inequality \cite{cglmp}. We can rewrite this inequality after symmetry transformations and mapping it into our normalized probability space $\mathcal{P}$ to obtain,
\begin{eqnarray}
\mathcal{C}_{CGLMP}=d\times p(\sum_{j=1}^{2}m_{j}=1|0,0)- & \nonumber \\
\sum_{\textbf{s}}(-1)^{s_{1}+s_{2}}p(\sum_{j=1}^{2}m_{j}=1|\textbf{s})+& \nonumber \\
\sum_{\textbf{s}}(-1)^{s_{1}+s_{2}}\sum_{k=2}^{d-1}(d-k-1)p(\sum_{j=1}^{2}m_{j}=k|\textbf{s})&\leq d. \nonumber \\
\end{eqnarray}
This inequality is also equivalent to the CHSH inequality for $d=2$ which we now write in the space $\mathcal{P}$ for clarity:
\begin{equation}
\mathcal{C}_{d=2}=\sum_{\textbf{s}}(-1)^{s_{1}s_{2}}p(\sum_{j=1}^{2}m_{j}=1|\textbf{s})\leq 2.
\end{equation}
For all possible correlators in $\mathcal{P}$, the maximal value of the left-hand-side of the CGLMP inequality is $2d-1$, thus violating it. In fact, for all $d$, this maximal violation of the CGLMP is obtained by a vertex of $\mathcal{P}$ corresponding to the map $f(\textbf{s})=s_{1}s_{2}+1$, i.e. the correlator $p(k|\textbf{s})=\delta^{k}_{s_{1}s_{2}+1}$. Masanes has shown that the only non-trivial, facet Bell inequalities are those that are equivalent to the CGLMP for the $(2,2,3)$ setting \cite{masanes}; this is confirmed by Tab. \ref{tab:tab2}.

For the $(2,2,2)$ setting there are $2^{4}-2^{3}=8$ non-$n$-partite linear functions and also $8$ inequalities in the symmetry class of the CHSH inequality. This is no coincidence as every Bell inequality in the symmetry class in maximally violated by a vertex of $\mathcal{P}$ corresponding to a non-$n$-partite linear function. This also occurs for the $(2,2,3)$ setting, there are $3^{4}-3^{3}=54$ non-$n$-partite linear functions and $54$ inequalities in the symmetry class of the CGLMP inequality. It can also be checked that every inequality in this symmetry class is violated by a different non-$n$-partite linear function.

The correspondence between non-$n$-partite linear functions and facet Bell inequality echoes the non-trivial Bell inequalities described by the inequalities of the form \eqref{genchsh}. There is a computational aspect to each inequality as a violation indicates that a correlator is in some sense closer to the non-$n$-partite linear functions. However, the correspondence breaks down for the $(2,2,5)$ setting where there are $5^{4}-5^{3}=500$ non-$n$-partite linear functions but $1000$ non-trivial, facet Bell inequalities in total. In this instance, each vertex of $\mathcal{P}$ maximally violates two facet Bell inequalities (each belonging to different symmetry classes). Therefore there are $1000$ facet Bell inequalities (see Appendix \ref{app1} for more details).

\textit{Tripartite Facet Bell Inequalities} - We have given an indication that facet Bell inequalities for $n=c=2$ have a computational interpretation. Every facet Bell inequality we have found is maximally violated uniquely by a vertex of $\mathcal{P}$. In this sense the violation of a facet Bell inequalities can quantify how computationally powerful a theory is. For situations with $n>2$, this becomes more complicated even for $n=3$ and $c=d=2$. The Mermin inequality \cite{mermin} for example is in the symmetry class of the inequality
\begin{equation}\label{mermin}
\sum_{\textbf{s}}\delta^{s_{1}}_{s_{2}}(-1)^{s_{1}s_{3}}p(\sum_{j=1}^{3}m_{j}=1|\textbf{s})\leq 2.
\end{equation}
Two vertices of $\mathcal{P}$ that are not vertices of $\mathcal{L}$ maximally violate this inequality. The two maps that correspond to these vertices are $f(\textbf{s})=s_{1}s_{2}s_{3}+1$ and $f(\textbf{s})=s_{1}s_{3}+1$. Therefore there is no longer a one-to-one correspondence between facet inequalities and the vertices of $\mathcal{P}$ that violate the inequality. 

The CHSH inequality expressed in terms of the expectation value of measurements (see next subsection) for two parties can be used to build facet Bell inequalities for more parties \cite{ww} such as the above inequality \eqref{mermin}. Analogously, we define CGLMP for three parties using the two party inequality. We have three parties but now we only consider non-zero terms in a Bell inequality when the third party's measurement setting $s_{3}=0$. For LHV correlators $p(k|s_{1},s_{2},0)$ the $n$-partite linear functions that can be achieved are $f(\textbf{s})=[\alpha_{1}s_{1}+\alpha_{2}s_{2}+\alpha_{3}]_{3}$ with $\alpha_{1},\alpha_{2},\alpha_{3}\in\mathbb{Z}_{d}$: the $n$-partite linear functions on two variables $s_{1}$ and $s_{2}$. Since the CGLMP inequality is facet-defining for the region of LHV correlators for two parties, or variables $s_{1}$ and $s_{2}$, it is facet-defining for this space of the $n=3$ correlators for $s_{3}=0$. Then we can write the tripartite CGLMP inequality as
\begin{eqnarray}
\mathcal{C}^{'}_{CGLMP}=d\times p(\sum_{j=1}^{3}m_{j}=1|0,0,0)- & \nonumber \\
\sum_{\textbf{s}}(-1)^{s_{1}+s_{2}}p(\sum_{j=1}^{3}m_{j}=1|s_{1},s_{2},0)+& \nonumber \\
\sum_{\textbf{s}}(-1)^{s_{1}+s_{2}}\sum_{k=2}^{d-1}(d-k-1)p(\sum_{j=1}^{3}m_{j}=k|s_{1},s_{2},0)&\leq d. \nonumber \\
\end{eqnarray}
For the case of $(3,2,3)$, this tripartite CGLMP inequality is facet-defining and forms a symmetry class with $324$ inequalities. This is also true for $d=2$ where $\mathcal{C}'_{CGLMP}$ results from the CHSH inequality and forms one of $4$ non-trivial symmetry classes. Also in this case there is another way of substituting the CHSH inequality to obtain the class containing the inequality \eqref{mermin}. This class also generalizes to the $d=3$ setting so that the following inequality is against the space where $s_{1}=s_{2}$ in a similar fashion to $\mathcal{C}^{'}_{CGLMP}$:
\begin{eqnarray}
\mathcal{C}^{''}_{CGLMP}=d\times p(\sum_{j=1}^{3}m_{j}=1|0,0,0)- & \nonumber \\
\sum_{\textbf{s}}\delta^{s_{1}}_{s_{2}}(-1)^{s_{1}+s_{3}}p(\sum_{j=1}^{3}m_{j}=1|\textbf{s})+& \nonumber \\
\sum_{\textbf{s}}\delta^{s_{1}}_{s_{2}}(-1)^{s_{1}+s_{3}}\sum_{k=2}^{d-1}(d-k-1)p(\sum_{j=1}^{3}m_{j}=k|\textbf{s})&\leq d. \nonumber \\
\end{eqnarray}
This inequality also forms a symmetry class with $324$ inequalities. There are $60$ other symmetry classes for $(3,2,3)$ (excluding the class of positivity and normalization inequalities). Inequalities from each of these classes can be found in the supplementary material \cite{website}.

\textit{Bipartite Three-Setting, Two-Outcome Facet Inequalities} - Finally, as can be seen from Tab. \ref{tab:tab2} for the $(2,3,2)$ setting there is only one symmetry class of non-trivial facet Bell inequalities. The Bell inequality generating this symmetry class is a generalization of the CHSH inequality:
\begin{equation}
\mathcal{C}_{c=3}=\sum_{\textbf{s}}(-1)^{s_{1}s_{2}}\prod_{j=1}^{2}(\delta^{s_{j}}_{0}+\delta^{s_{j}}_{1})p(\sum_{j=1}^{2}m_{j}=1|\textbf{s})\leq 2.
\end{equation}
This is exactly the same as the CHSH inequality if instead of $c=3$ we had $c=2$. In fact for either $c$ or $d$ equal to $4$ and $n=2$, inequalities constructed from the CHSH inequality capture a lot of the structure of the LHV polytope. 

In this subsection we have given an indication of the richness of the structure of the LHV polytope for some simple settings. There is also a computational element to some of these Bell inequalities, such as the CGLMP inequality which is maximally violated by a vertex of $\mathcal{P}$, in that greater violations enable correlators to come closer to evaluating a non-$n$-partite linear function deterministically. Quantum correlators are known to provide violations of Bell inequalities and so in the next section we discuss quantum correlators. 

\subsection{Quantum correlators}\label{sec3c}

\noindent
Quantum correlations can violate all manner of Bell inequalities if the correlations are generated from measurements on an entangled state. However, the connection between a violation of a facet Bell inequality and entanglement is not completely clear. Recently the two concepts have become very distinct and nowhere is this best demonstrated by the CGLMP inequality where the greatest violation of the CGLMP inequality is \emph{not} achieved with a bipartite maximally entangled state. The region $\mathcal{Q}\subset\mathcal{P}$ of quantum correlators is convex, but it is not a polytope since there are an infinite number extreme points \cite{pitowski}. It is not clear in general how to define all of the extreme points of $\mathcal{Q}$. We now discuss this region in the context of our CGLMP-like Bell tests.

Without loss of generality, we can assume that the measurements at each site are projective measurements by Naimark's theorem \cite{naimark}, so that for a quantum state $\rho\in(\mathcal{H})^{\otimes n}$, the correlators are
\begin{align}
p_{\mathcal{Q}}(k|\textbf{s}) = \sum_{m}\delta^{\sum_{j=1}^{n}m_{j}}_{k}\tr(\rho\bigotimes_{j=1}^{n}|m_{j}\rangle_{s_{j}}\langle m_{j}|_{s_{j}}).
\end{align}
As $\rho$ can be expressed as a convex combination of pure states, we can, for the purposes of finding the maximum quantum violation, assume that $\rho$ is a pure state.

A more compact way of expressing Bell inequalities is in terms of the expectation values of joint measurements $M_{s_j}$ \cite{beamsplitter, general}, which, for the above projective measurements, are
\begin{eqnarray}\label{qbell}
\mathbb{E}(\textbf{s})&=&\tr(\rho\bigotimes_{j=1}^{n}M_{s_{j}})\nonumber \\
&=&\sum_{m}e^{\frac{i2\pi}{d}(\sum_{j=1}^{n}m_{j})}\tr(\rho\bigotimes_{j=1}^{n}|m_{j}\rangle_{s_{j}}\langle m_{j}|_{s_{j}}) \nonumber \\
&=&\sum_{k=0}^{d-1}e^{\frac{i2\pi k}{d}}p_{\mathcal{Q}}(k|\textbf{s}).
\end{eqnarray}
The correlators $p_{\mathcal{Q}}(k|\textbf{s})$ in \eqref{qbell} can be replaced with LHV correlators to obtain the expectation value of measurements for LHV theories.

Our construction of correlators therefore has a natural role in the construction of expectation values $\mathbb{E}(\textbf{s})$. If we have a Bell inequality in terms of correlators such as \eqref{bell}, it is possible to relate it to a Bell inequality in terms of these expectation values by taking the discrete Fourier transform \cite{general}. The discrete Fourier transform of \eqref{bell} gives
\begin{equation}\label{qbell2}
\sum_{\textbf{s}}\sum_{\mu=0}^{d-1}\eta_{\mu,\textbf{s}}[\mathbb{E}(\textbf{s})]^{\mu}\leq \gamma_{\mathcal{L}}
\end{equation}
with $\gamma_{\mathcal{L}}$ as defined in \eqref{bell} if the complex pre-factors $\eta_{\mu,\textbf{s}}$ are
\begin{equation}
\eta_{\mu,\textbf{s}}=\frac{1}{d}\sum_{k=1}^{d-2}\omega_{k,\textbf{s}}e^{-i2\pi\frac{\mu k}{d}}.
\end{equation}
This construction is another motivation for considering correlators as opposed to the full distribution $p(\textbf{m}|\textbf{s})$.

To obtain the maximal violation of a Bell inequality by quantum correlators, we optimize \eqref{qbell2} over pure states $\rho$ and unitary operators $M_{s_{j}}$ (which correspond to projective measurements). For the case of $(n,2,2)$, all facet Bell inequalities are maximally violated by the GHZ state $|GHZ\rangle=\frac{1}{\sqrt{2}}(|0\rangle^{\otimes n}+|1\rangle^{\otimes n})$ \cite{ww}. In general, finding the maximal violation is a difficult problem, however, there are methods of providing numerical lower and upper bounds on this quantum violation (see e.g. \cite{navascues} and \cite{beamsplitter}). We will now discuss and utilize these methods to find the maximal quantum violations of facet Bell inequalities for two parties.

First we discuss methods of finding a lower bound on a two-party maximal quantum violation used in \cite{beamsplitter} and \cite{acin}. The quantum state is first fixed as the $d^{2}$-dimensional maximally entangled state $|\Psi\rangle=\frac{1}{\sqrt{d}}\sum_{j=0}^{d-1}|jj\rangle$ and we optimize over the unitaries $M_{s_{j}}$. More specifically we write the projectors as $|m_{j}\rangle_{s_{j}}\langle m_{j}|_{s_{j}}=M_{s_{j}}|k\rangle\langle k|M_{s_{j}}^{\dagger}$, where $\{|k\rangle|k\in\mathbb{Z}_{d}\}$ is the standard basis of $\mathcal{H}^D$. The $M_{s_{j}}$ can be written as $M_{s_{j}}=FD_{s_{j}}$ where $F$ is the $d$-by-$d$ Hadamard, or Quantum Fourier Transform matrix and $D_{s_{j}}=\textrm{diag}(e^{i\phi_{1}(s_{j})},e^{i\phi_{2}(s_{j})},...,e^{i\phi_{d}(s_{j})})$, a diagonal matrix with $\phi_{j}(s_{j})$ as real phases. Therefore we optimise over these phases $\phi_{j}(s_{j})$ to numerically maximize the quantum violation for the maximally entangled state.

\begin{table}
\begin{center}
 \begin{tabular}{| c | c | c | c | c | c | c |}
 \hline
 n & c & d & Symmetry class & LHV bound & Quantum bound\\ \hline
 2 & 2 & 2 & $\mathcal{C}_{d=2}$ & 2 & $2.4142^{\dagger}$ \\ \hline
 2 & 2 & 3 & $\mathcal{C}_{CGLMP}$ & 3 & $3.9149$ \\ \hline
 2 & 2 & 5 & $\mathcal{I}_{1}$ & 5 & $6.3145$\\
 2 & 2 & 5 & $\mathcal{I}_{2}$ & 5 & $7.6290$\\
 2 & 2 & 5 & $\mathcal{I}_{3}$ & 5 & $7.0314$\\
 2 & 2 & 5 & $\mathcal{C}_{CGLMP}$ & 5 & $7.0314$\\ \hline
 2 & 3 & 2 & $\mathcal{C}_{c=3}$ & 2 & $2.4142^{\dagger}$ \\
 \hline
 \end{tabular}
\end{center}
\caption{We list the bipartite maximal quantum violations for particular facet Bell inequalities for $c$ and $d$. The inequalities $\mathcal{I}_{1}$, $\mathcal{I}_{2}$ and $\mathcal{I}_{3}$ can be found in the appendix \ref{app1}. Those violations that are achieved with the bipartite maximally entangled state of $d^2$ dimension are labeled with a $\dagger$.}
\label{tab:tab3}
\end{table}

Once we find the optimal phases $\phi_{j}(s_{j})$ for each $s_{j}$, we optimize over the pure states, $\rho$, by finding the largest eigenvalue of the operators corresponding to \eqref{qbell2}. The largest eigenvalue then corresponds to the eigenvector $|\psi\rangle$. Using this method it has been shown that the maximal violation of the CGLMP inequality is achieved with a state that is not maximally entangled \cite{acin}.

One method of providing an upper bound to the violation of a Bell inequality is the use of semi-definite programming \cite{navascues}. A Gram matrix is constructed from the expectation value of products (or sequences) of projectors and this matrix is positive semi-definite. The Bell inequality is then a linear objective function on this matrix subject to linear constraints on elements of the Gram matrix. Navascues et al constructed a hierarchy of semi-definite programs that converge to the set of quantum correlations and also constructed a method for testing for convergence \cite{navascues}.

We used YALMIP and SeDuMi \cite{yalmip} to implement these semi-definite algorithms to find the upper bound to quantum violations whilst finding lower bounds using previously described methods. In Tab. \ref{tab:tab3}, we list the maximal quantum violations of the facet Bell inequalities in symmetry classes described previously; a single value is listed as the upper and lower bounds coincide differ by at most $10^{-9}$, which is consistent with numerical error. Interestingly, of the settings studied, only Bell inequalities for qubits are maximally violated by a maximally entangled bipartite state.

\subsection{Summary}\label{sec3d}

\noindent
Dual to a vertex description of a convex polytope is the facet description where facets are defined by linear inequalities. We have given an insight into some of the complicated structure behind these facet Bell inequalities. Despite the complicated structure, there is a computational insight into the facet Bell inequalities for particular settings. The CGLMP inequalities are computational in nature and can be used to construct facet Bell inequalities for multipartite scenarios. We then showed that these and other bipartite facet Bell inequalities are violated by quantum correlators, which indicates that quantum mechanics offers a computational advantage relative to any LHV theory.

The non-trivial Bell inequalities that we have constructed in this section have an explicit computational meaning. We will explore these inequalities in the context of non-local games and Measurement-based Quantum Computing. More generally we consider how much of the structures developed in this paper can be applied to these information processing scenarios. Whilst a lot of ideas have been generalized from the $(n,2,2)$ setting, we show that LHV correlators are more powerful when $c,d>2$. 

\section{Bell tests, Non-local Games and Quantum Computing}\label{sec4}

\noindent
Connections have been made between Bell tests and Measurement-based Quantum Computing (MBQC) \cite{anders, hoban2}. This connection has been explicitly explored in the setting where each party has a choice of two measurements, each with two outcomes. Furthermore, the role of post-selection in Bell tests simulating MBQC has been explored in this setting \cite{hoban} leading to novel quantum phenomena. In this section we will give an overview of all of these connections and then discuss their implementation in general $(n,c,d)$ settings. So much of the computational insight in constructing LHV correlators in the $(n,2,2)$ setting carries over into the $(n,c,d)$ setting, but the structure of LHV correlators is also richer. This richness also means certain results do not generalize. 

We have already discussed non-local games in the context of the non-trivial Bell inequalities in section \ref{sec3}. There is an overlap between non-local games and MBQC if we consider the elements in both information processing scenarios. Firstly in MBQC there are a number of sites that share a particular resource (e.g. the cluster state \cite{mbqc} for quantum computing), and single-site measurements are made on this resource. All measurement data is processed by a classical computer and in order to achieve a universal quantum computer in current models, adaptive measurements are required; adaptivity means choices of measurements are informed by previous measurement outcomes. In the model developed by Raussendorf and Briegel, the classical computer only needs to be able to perform linear Boolean functions to achieve a universal quantum computer \cite{mbqc}. 

We now recall the model of non-local games as discussed earlier \cite{non-localgame}. There are a number of parties who share some resource state but do not communicate with each other. These parties receive an input from a referee and send an output back. The referee processes the outputs to see if the parties successfully performed some task. MBQC can be recast as a non-local game where the parties choose their measurement based upon the input they receive and send the referee the measurement outcome. The referee in MBQC is a classical computer who processes the measurement data in order to achieve some task. Therefore the referee performs linear Boolean function computation on the measurement data described by bits \cite{hoban2}. 

MBQC has been generalized to include more than a choice of two measurements and two measurement outcomes at each site \cite{highdemmbqc}, thus measurement data is no longer encoded in bits. The classical computer processing measurement data in these models just uses addition modulo $d$ where $d$ is the number of measurement outcomes at each site \cite{highdemmbqc}. Addition modulo $d$ on data is still an extremely limited form of computation, even compared to the $n$-partite linear functions discussed in this paper. 

For simplicity we say that the data sent and received from the measurements sites is encoded as digits in $\mathbb{Z}_{d}$. The specific model of MBQC as a non-local game that we consider consists of three stages:
\begin{enumerate}
\item A digit-string $\textbf{x}\in\mathbb{Z}^{|\textbf{x}|}_{d}$ of length $|\textbf{x}|$ is processed by a classical computer, which then sends a single digit to each site after processing;
\item At each site, this digit is used to select a measurement and the outcome is sent to the classical computer;
\item The classical computer processes this measurement data to output a single digit. 
\end{enumerate}
The processing power of the classical computer is limited to addition modulo $d$. 

The input at each site $s_{j}=h(\textbf{x})$ is now a result of some pre-processing leading to a function $h(\textbf{x})$ on the digit-string $\textbf{x}$ where $|\textbf{x}|\leq n$. This function is limited to addition modulo $d$ on $\textbf{x}$, i.e. $h(\textbf{x})=\sum_{j=1}^{|\textbf{x}|}\alpha_{j}\textbf{x}_{j}$ where $\alpha_{j}\in\mathbb{Z}_{d}$. Then the processing on measurement outcomes $m_{j}$ leads to the output function $\sum_{j=1}^{n}m_{j}$ without loss of generality. 

The non-trivial Bell inequalities described in this paper capture the mean success probability of LHV correlators evaluating a non-$n$-partite linear function $f(\textbf{s})$ for the input $\textbf{s}$. We now investigate how the construction of Bell inequalities is modified by this new element of pre-processing on a digit-string $\textbf{x}$. The pre-processing allows us to express the correlators as being conditioned upon $\textbf{x}$, i.e., $p(k|\textbf{x}):=p(k|\textbf{s}(\textbf{x}))$.

The correlators $p(k|\textbf{x})$ are now elements of a vector $\vec{p}$ in $\mathbb{R}^{(d-1)d^{|\textbf{x}|}}$. We now describe the structure of the region of all possible correlators.

For the digit-string $\textbf{x}$ if we set $n=|\textbf{x}|$ and $s_{j}=x_{j}$ for all $j$, then the possible correlators $p(k|\textbf{x})$ live in the space $\mathcal{P}$ for this setting. If we increase $n$ and consider other forms of linear pre-processing, we allow more freedom to evaluate more complicated functions. As shown in \cite{hoban2} for the pre-processing described above, quantum correlators for $d=2=c$ benefit from this freedom. That is, for $d=2=c$ settings, all deterministic correlators $p(k|\textbf{x})$ corresponding to functions $f(\textbf{x}):\mathbb{Z}_{d}^{|\textbf{x}|}\rightarrow\mathbb{Z}_{d}$ can be achieved for some sufficiently large $n$.

Remarkably, LHV correlators $p(k|\textbf{x})$ for $d=2$ are unaffected by increasing $n$ using the pre-processing described above \cite{hoban}. The deterministic correlators, or vertices of $\mathcal{P}$ for $p(k|\textbf{x})$ that LHV correlators achieve are always the linear Boolean functions. In this way, the computational power of LHV correlators is not boosted by increasing $n$ and pre-processing for $d=2$. We now show that for general $d>2$ this no longer holds and computational power can be boosted. We actually show something stronger: if any form of the pre-processing described above with $s_{j}=h(\textbf{x})$ and $|\textbf{x}|=n$ is allowed, then the computational power of LHV correlators will be boosted.

\begin{prop}\label{prop:pre-processing}
For arbitrary addition modulo $d$ pre-processing on $\textbf{x}$ leading to measurement settings $s_{j}=h(\textbf{x})$, the space of LHV correlators $p(k|\textbf{x})$ is not confined to the convex hull of $n$-partite linear functions on $x$ for $d>2$.
\end{prop}

\begin{proof}
As before, we consider the deterministic maps $p(k|\textbf{s})$ and then take their convex hull. However, we now consider the effect of pre-processing on $\textbf{x}$ and make the assumption that $d>2$. When each party receives the input $s_{j}$ generated by the pre-processing, the input is $s_{j}=h(\textbf{x})$ where $h(\textbf{x})=\sum_{j=1}^{|\textbf{x}|}\alpha_{j}x_{j}$ with $\alpha_{j}\in\mathbb{Z}_{d}$. We now map into the space of all correlators $p(k|\textbf{x})$ under the influence of all possible pre-processing of the form $h(\textbf{x})$. 

For LHV correlators, we need to consider all single site maps $\mathbb{Z}_{d}\rightarrow\mathbb{Z}_{d}$ which can be written as a polynomial over $\mathbb{Z}_{d}$ as:
\begin{equation}\label{mapproc}
m_{j}=\sum_{y=1}^{(d-1)}\sum_{z=0}^{(d-1)}\epsilon_{y,z}s_{j}^{d-(z+1)}+\gamma,
\end{equation}
where $\epsilon_{y,z}=\beta_{y}(-1)^{z+1}{{d-1}\choose{z}}y^{z}$ and $\gamma$, $\beta_{y}\in\mathbb{Z}_{d}$. If we add in the pre-processing stage where $s_{j}=\sum_{j=1}^{|x|}\alpha_{j}x_{j}$ then single site maps become polynomials in elements of $x_{j}$, i.e. $s_{j}^{d-(z+1)}=(\sum_{j=1}^{|x|}\alpha_{j}x_{j})^{d-(z+1)}$. Therefore for appropriately chosen $\beta_{y}$ and $\alpha_{j}$, there are now cross-terms between elements of $x_{j}$, e.g. $x_{1}x_{2}$ etc if $d>2$. Because of this the deterministic single-site maps with pre-processing cannot be described as $n$-partite linear functions on $\textbf{x}$. 
\end{proof}

If we just have $s_{j}=x_{j}$ and have $n=|\textbf{x}|$ parties then the LHV polytope is just the convex hull of $n$-partite linear functions on $\textbf{x}$. Pre-processing therefore can boost the computational power of LHV correlators to go beyond this computational description. However, not all polynomials of elements $x_{j}$ over $\mathbb{Z}_{d}$ can be achieved by taking powers of the linear functions $\sum_{j=1}^{|\textbf{x}|}\alpha_{j}x_{j}$. For example, the function $f(\textbf{x})=\prod_{j=1}^{|\textbf{x}|}x_{j}^{d-1}$ cannot be achieved by this method. More generally, any function that contains a term that is a non-$n$-partite linear function of degree greater than $d-1$ cannot be achieved in any LHV theory for any value of $n$. So whilst the LHV correlators are boosted by this pre-processing, the resulting LHV polytope does not encompass the convex hull of all functions $f(\textbf{x})$. 

Knowing that the boosted LHV polytope cannot achieve a function such as $f(\textbf{x})=\prod_{j=1}^{|\textbf{x}|}x_{j}^{d-1}$, we can construct a non-trivial Bell inequality for this function for all $n\geq|\textbf{x}|$. Instead of finding the upper bound for LHV correlators by taking the supremum over all $n$-partite linear functions, we now have to take the supremum over all possible linear pre-processing for all $n$ maps of the form \eqref{mapproc}. This added complication highlights the uniqueness and simplicity of the case for $d=2$, where one only needs to take the supremum over all linear Boolean functions of $\textbf{x}$ even with pre-processing.

In the case where the number of possible inputs at each site is different from the number of possible outputs, there are other complications. We illustrate this with an example in the $(3,2,3)$ setting but with a bit-string $\textbf{x}\in\mathbb{Z}^{2}_{2}$. We have three parties with pre-processing on $\textbf{x}$ leading to inputs: $s_{1}=x_{1}$, $s_{2}=x_{2}$ and $s_{3}=\left[x_{1}+x_{2}\right]_{2}$. With this pre-processing one can achieve the non-$n$-partite linear function $f(\textbf{x})=x_{1}x_{2}$ with the $n$-partite linear function $f(\textbf{s})=2s_{1}+ 2s_{2} + s_{3}$ and this pre-processing as $2x_{1}+2x_{2}+[x_{1}+x_{2}]_{2}=\delta^{x_{1}x_{2}}_{1}$. Therefore even with $c=2$ and $d>2$, the LHV polytope is also no longer confined to $n$-partite linear functions on $\textbf{x}$.

The effect of processing on LHV correlators in the general framework of Bell tests is pointedly different from the two-measurements, two-outcome case. We hope to illustrate that there is a richness in structure in how information processing affects classical correlations and an appreciation of the connections between computation and correlators to give a handle on the effect of processing. In \cite{hoban}, the effect of processing can be seen to have a role in creating so-called loopholes in Bell tests and it is worthwhile exploring loopholes in the framework of data processing described here. 

We also hope these results really emphasize why the $(n,2,2)$ setting is so special and such a powerful platform for distinguishing between classical and non-classical correlators. Even with linear Boolean pre-processing, the classical computation possible with LHV correlators is well characterized. This is still a computationally interesting setting as two inputs and two outputs at each site already enables quantum computing and the richness of structure that comes with that. Despite the fact that MBQC can be generalized to the setting with more inputs and outputs at each site, the same LHV structure does not retain the same character in this setting. 

\section{Summary and Open Questions}\label{sec5}

\noindent
In this paper, we have explored the relatively under-studied multi-setting, multi-outcome Bell inequality test in the context of the computational perspective of Bell tests, which has allowed us to derive new families of non-trivial Bell inequalities. We have presented our work in terms of correlators and not the full measurement statistics of an experiment for clarity and convenience. However, we have shown that a significant amount of the structure of the space of non-signaling probability distributions are captured by this setting. 

We have shown that LHV theories can be generally considered in terms of limited computational expressiveness, i.e. LHV theories can only evaluate $n$-partite linear functions as defined. A violation of a Bell inequality defined by these functions implies a computational advantage relative to any LHV theory. Given this interpretation, we can construct an infinite class of non-trivial Bell inequalities. Furthermore we have explicitly constructed facet Bell inequalities that define the region of LHV correlators using the construction of LHV correlators in this paper. 

Central to this paper is the incorporation of a computational perspective using modulo arithmetic into the more established picture of the convex polytopes of LHV theories. This perspective provides a broad arena to explore and, while we have given some insight, many open questions still remain. Firstly, what is the general correspondence between non-$n$-partite linear functions and facet Bell inequalities? There seems to be a one-to-one or a one-to-two correspondence in certain settings but this is not always the case. There are many happy coincidences that need further clarification.

A huge open question is what exactly defines the quantum region in our space of correlators? For the $(n,2,2)$ correlator Bell test, we know that measurements on GHZ states form the extreme points of the quantum region \cite{ww,zb}. However, little is known for other settings. At the moment finding the eigenvectors for maximal violations of a Bell inequality has not resulted in a general expression for the states that maximally violate the CGLMP inequality, if such a closed form expression even exists. In addition, more study is needed to determine the effect of pre-processing on quantum correlators.

In \cite{hoban}, signaling correlators were simulated with post-selection on measurement data. We have shown that the generalization of this post-selection leads to a boost in the power of LHV correlators; something that is akin to a loophole in a Bell test. However, in the more general setting, does simulating signaling correlators also boost LHV correlators? Are there more nuanced ways of processing data that does not expand the region of LHV correlators as occurs in Proposition \ref{prop:pre-processing}?

If a violation of a Bell inequality is a resource for information processing \cite{qkd, random, comm, anders, hoban2}, then all aspects of Bell tests need to be explored to determine how we can best exploit this resource. We have tried to hint at how the richness of structure of abstract local hidden variable theories can be captured by a simple computational picture. We hope that this work gives more insight into the cooperative relationship between quantum information and the foundations of quantum mechanics. 

\textit{Acknowledgements} - We would like to thank Hussain Anwar, Nicolas Brunner, Miguel Navascu\'{e}s and Stephen Bartlett for interesting discussions. MJH is financially supported by the EPSRC (UK) and DEB acknowledges the support of a Leverhulme Fellowship.

\begin{appendix}
\section{Facet Bell inequalities}\label{app1}

\noindent
In this section we list the facet Bell inequalities which belong to symmetry classes that are not explicitly described in the paper. We express the correlators $p(k|\textbf{s})$ as the elements of a $(d-1)c^{n}$ vector in real space $\mathbb{R}^{(d-1)c^{n}}$. To simplify we describe the inequalities in terms of a vector $\vec{b}$ that has inner product with correlator column vectors $\vec{p}$ to form the Bell inequality, i.e. $\vec{b}\cdot\vec{p}\leq\gamma_{\mathcal{L}}$. The coefficients correspond to elements going from left-to-right as $p(1|0,0,...,0)$, $p(2|0,0,...,0)$ up to $p(d-1|c-1,c-1,...,c-1)$. We will use this notation in the next two sections of the appendix. For the $(2,2,5)$ setting, there are $4$ symmetry classes of facet Bell inequalities (excluding the class of positivity and normalization inequalities). One of these is the CGLMP inequality and the other three are given by the vectors of coefficients
\begin{align}
\vec{b}_1 &= \frac{1}{2}\left(6,2,3,4,4,-2,2,1,4,-2,2,1,-4,2,-2,-1\right),	\nonumber\\
\vec{b}_2 &= \left(3,1,-1,-3,2,-1,-4,-2,2,-1,-4,-2,-2,1,4,2\right),	\nonumber\\
\vec{b}_3 &= \left(2,-1,1,-2,3,1,-1,2,3,1,-1,2,-3,-1,1,-2\right).
\end{align}
Denoting by $\mathcal{I}_j$ the inequality corresponding to $\vec{b}_j$, $\mathcal{I}_1$ and the CGLMP inequality are maximally violated by the vertex corresponding to $f(\textbf{s})=s_{1}s_{2}+1$ and $\mathcal{I}_2$ and $\mathcal{I}_3$ are maximally violated by the vertex corresponding to $f(\textbf{s})=2s_{1}s_{2}+1$.

The LHV bounds for all of the inequalities are listed in Tab. \ref{tab:tab3} with their quantum violation.

As we have seen there is a corresponding function for each of these inequalities that leads to a maximal violation. Each of these functions is a vertex in the general space of Bell correlators mentioned in the paper.

\section{Generalization to non-prime number of settings}\label{app2}

\noindent
In this paper, we focused on the case where $c$ and $d$ are prime. However, this assumption is not needed to prove all the results. We now give a quick overview of how we describe LHV correlators in terms of polynomials on inputs for non-prime $c$. For non-prime number of settings, the number of inputs at each site $c$ can be just expressed as the Cartesian product of the prime factors of $c$, i.e. $\mathbb{Z}_{c}\rightarrow \mathbb{Z}_{c_{1}}\times \mathbb{Z}_{c_{2}}\times ...\times \mathbb{Z}_{c_{q}}$ where $\{c_{1}, c_{2},...,c_{q}\}$ is the set of prime factors of $c$ and $q$ is the number of prime factors in $c$. As a result each input $s_{j}$ can be represented as a string of elements $s_{j}=\{t_{j}^{1},t_{j}^{2},...,t_{j}^{q}\}$ where $t_{j}^{k}\in \mathbb{Z}_{c_{k}}$.

For LHV theories we consider the single site maps $m_{j}$ and then take their sum modulo $d$. For non-prime dimension then, each deterministic single site map can be written as
\begin{equation}
m_{j}=\alpha_{j}(\lambda)+\sum_{\underset{v\neq\textbf{0}}{v\in\mathbb{Z}_{c},}}\beta_{j}^{v}(\lambda)\delta^{s}_{v},
\end{equation}
where $\alpha_{j}(\lambda)\in\mathbb{Z}_{d}$ and $\beta_{j}^{v}(\lambda)\in\mathbb{Z}_{d}$ dependent on the local hidden variable $\lambda$ and $\textbf{0}=\{0, 0,...,0)\}$. We can use the redefinition of the delta function over $\mathbb{Z}_{c}\rightarrow \mathbb{Z}_{c_{1}}\times \mathbb{Z}_{c_{2}}\times ...\times \mathbb{Z}_{c_{q}}$ used for each prime dimension $\mathbb{Z}_{c_{k}}$. This results in the map:
\begin{equation}
m_{j}=\alpha_{j}(\lambda)+\sum_{\underset{v\neq\textbf{0}}{v\in\mathbb{Z}_{c},}}\beta_{j}^{v}(\lambda)\prod_{k=1}^{q}\left[1-(t_{j}^{k}-v_{j}^{k})^{(c_{k}-1)}\right]_{c_{k}}.
\end{equation}
One then takes the sum modulo $d$ of all of these single site maps to obtain (a superficially complicated) map which is $n$-partite linear as it only consists of single-site maps. There is also no multiplication between values of $t_{j}^{k}$ with other values of $t_{j'}^{k'}$ for $j\neq j'$. Again there are $d^{n(c-1)+1}$ possible deterministic maps for LHV theories, and their convex hull forms the LHV polytope $\mathcal{L}$.

\begin{table}
\begin{center}
 \begin{tabular}{| c | c c c c c c c c c c c c c c c c |}
 \hline
 & & & & & & & & $\vec{b}$ & & & & & & & & \\ \hline
 $\mathcal{B}_{1}$& 2&2&1&1&2&-1&-1&-2&1&-1&-2&2&1&-2&2&1\\
 $\mathcal{B}_{2}$& 2&2&1&1&2&-1&-1&-2&1&-2&2&1&1&-1&-2&2\\
 $\mathcal{B}_{3}$& 2&2&1&1&2&-1&-2&-1&1&-2&1&2&1&-1&2&-2\\
 $\mathcal{B}_{4}$& 2&2&1&1&1&-1&2&-2&1&-2&1&2&2&-1&-2&-1\\
 $\mathcal{B}_{5}$& 2&2&1&1&1&-2&2&1&1&-1&-2&2&2&-1&-1&-2\\
 $\mathcal{B}_{6}$& 2&1&1&0&1&-1&-1&1&1&-1&-1&-1&0&1&-1&0\\
 $\mathcal{B}_{7}$& 2&1&1&0&1&-1&-1&1&0&1&-1&0&1&-1&-1&-1\\
 $\mathcal{B}_{8}$& 2&1&1&0&0&1&-1&0&1&-1&-1&1&1&-1&-1&-1\\
 $\mathcal{B}_{9}$& 2&1&0&1&1&-1&1&-1&0&1&0&-1&1&-1&-1&-1\\
 $\mathcal{B}_{10}$& 2&1&0&1&0&1&0&-1&1&-1&1&-1&1&-1&-1&-1\\
 $\mathcal{B}_{11}$& 2&0&1&1&0&0&1&-1&1&1&-1&-1&1&-1&-1&-1\\
$\mathcal{C}^{1}_{c=4}$& 1&1&0&0&1&-1&0&0&0&0&0&0&0&0&0&0\\
$\mathcal{C}^{2}_{c=4}$& 1&1&0&0&0&0&0&0&1&-1&0&0&0&0&0&0\\
$\mathcal{C}^{3}_{c=4}$& 1&0&1&0&0&0&0&0&1&0&-1&0&0&0&0&0\\
 \hline
 \end{tabular}
\end{center}
\caption{The facet Bell inequality expressions that each belong to a particular symmetry class for $(2,4,2)$. Each row corresponds to a particular inequality belonging to a different symmetry class. Each column of $\vec{b}$ is an element of this vector that forms an inner product with $\vec{p}$. The LHV upper bound for inequalities $\mathcal{B}_{1}$, $\mathcal{B}_{2}$, $\mathcal{B}_{3}$, $\mathcal{B}_{4}$, and $\mathcal{B}_{5}$ is $8$ and $4$ for $\mathcal{B}_{6}$, $\mathcal{B}_{7}$, $\mathcal{B}_{8}$, $\mathcal{B}_{9}$, $\mathcal{B}_{10}$, and $\mathcal{B}_{11}$.}
\label{tab:tab4}
\end{table}

The smallest value of non-prime $d$ is $4$ and so we studied the structure of the LHV polytope for the $(n,4,2)$ setting which has $2^{7}$ $n$-partite linear functions corresponding to the vertices of the LHV polytope. There are $27,968$ facet Bell inequalities for this setting found with Polymake \cite{polymake} which reduces to $15$ symmetry classes ($14$ excluding the class of normalization and positivity inequalities). Three of these classes are forms of the CHSH inequality embedded in the larger number of inputs. For completeness, we have listed all $14$ Bell inequalities in Tab. \ref{tab:tab4}. We now explicitly write out one of these inequalities:
\begin{equation}
\mathcal{C}^{1}_{c=4}=\sum_{\textbf{s}}(-1)^{s_{1}s_{2}}\prod_{j=1}^{2}(\delta^{s_{j}}_{0}+\delta^{s_{j}}_{1})p(\sum_{j=1}^{2}m_{j}=1|\textbf{s})\leq 2.
\end{equation}
which is almost exactly the same as $\mathcal{C}^{1}_{c=3}$. The other two inequalities, $\mathcal{C}^{2}_{c=3}$ and $\mathcal{C}^{3}_{c=3}$ are similar to this inequality except with altered delta functions for $\mathcal{C}^{2}_{c=3}$:
\begin{equation}
\prod_{j=1}^{2}(\delta^{s_{j}}_{0}+\delta^{s_{j}}_{1})\rightarrow(\delta^{s_{1}}_{0}+\delta^{s_{1}}_{2})(\delta^{s_{2}}_{0}+\delta^{s_{2}}_{1}),
\end{equation}
and for $\mathcal{C}^{3}_{c=3}$:
\begin{equation}
\prod_{j=1}^{2}(\delta^{s_{j}}_{0}+\delta^{s_{j}}_{1})\rightarrow(\delta^{s_{1}}_{0}+\delta^{s_{1}}_{2})(\delta^{s_{2}}_{0}+\delta^{s_{2}}_{2}).
\end{equation}
%

Throughout the whole of the paper, we have assumed that $d$ has been prime. Indeed, this fact has been vital in the proof of Theorem \ref{thm:generalized-pr}, so this proof fails for non-prime $d$. However, for other aspects of the methods in this paper, a non-prime $d$ has little consequence for our construction of correlators. There are two approaches to considering non-prime $d$: the most obvious approach is just to think of representing outcomes as a string of digits from prime number registers as with the description of settings above; secondly, instead of having the elements $f(\textbf{s})$ of a function as elements of a vector in a vector space over $\mathbb{Z}_{d}$ for prime $d$, we now have elements of a \emph{module} over $\mathbb{Z}_{d}$; a generalization of a vector space. Since we only ever take addition of functions represented in a Kronecker delta function basis to generate all functions, all of our methods generalize to these modules. 

From the aspect of functions, the only difference for non-prime $d$ is the fact that $f(y)$ in \eqref{genfunc} can now take values from a non-prime register $\mathbb{Z}_{d}$. The derivation of all functions and $n$-partite linear functions generalizes naturally. For $(2,2,4)$ with the smallest non-prime $d=4$, we present the facet Bell inequalities. There are four symmetry classes for this setting of which one is the class of positivity and normalization inequalities. One of these three non-trivial symmetry classes is formed by the CGLMP inequality. The second is another generalization of the CHSH inequality and is of the form:
\begin{equation}
\sum_{\textbf{s}}(-1)^{s_{1}s_{2}}\left[p(\sum_{j=1}^{2}m_{j}=1|\textbf{s})+p(\sum_{j=1}^{2}m_{j}=3|\textbf{s})\right]\leq 2,
\end{equation}
which is essentially the CHSH inequality if each party groups their outcomes $m_{j}$ into modulo $2$ terms. In other words, since $1$ mod $2$ is equal to $3$ mod $2$, the above inequality is equivalent to the CHSH inequality if each party just maps from mod $4$ arithmetic to mod $2$.

The third and final symmetry class is generated by the following inequality (expressed in the notation described earlier):
\begin{equation}
\left(1,2,1,1,2,1,1,2,1,-1,-2,-1\right)\cdot \vec{p}\leq 4.
\end{equation}
The deterministic map or function $f(\textbf{s})$ that achieves the maximal upper bound $6$ of this inequality for all correlators in $\mathcal{P}$ is $f(\textbf{s})=2s_{1}s_{2}+2$. It is worth noting that this can be constructed by adding $\sum_{\textbf{s}}2(-1)^{s_{1}s_{2}}p(2|\textbf{s})$ to the previous inequality.

\end{appendix}

\end{document}